\documentclass[11pt]{article}

\usepackage{amssymb}
\usepackage{ textcomp }
\usepackage[margin=1in]{geometry}
\usepackage{complexity}
\usepackage{mathtools}
\usepackage{xspace}
\usepackage{rpmacros}
\RequirePackage[colorlinks=true]{hyperref}
\hypersetup{
  linkcolor=[rgb]{0,0,0.4},
  citecolor=[rgb]{0, 0.4, 0},
  urlcolor=[rgb]{0.6, 0, 0}
}
\usepackage{mathpazo}
\usepackage{bbm}
\usepackage{dsfont}
\usepackage{tikz}
\usetikzlibrary{arrows,shapes.geometric,positioning}

\usepackage{amsthm}
\usepackage{thmtools,thm-restate}

\numberwithin{equation}{section}
\declaretheoremstyle[bodyfont=\it,qed=\qedsymbol]{noproofstyle}

\declaretheorem[numberlike=equation]{observation}

\declaretheorem[name=Observation,numbered=no]{observation*}

\declaretheorem[numberlike=equation]{fact}

\declaretheorem[numberlike=equation]{problem}

\declaretheorem[numberlike=equation]{theorem}

\declaretheorem[name=Theorem,numbered=no]{theorem*}

\declaretheorem[numberlike=equation]{lemma}
\declaretheorem[name=Lemma,numbered=no]{lemma*}

\declaretheorem[numberlike=equation]{corollary}
\declaretheorem[name=Corollary,numbered=no]{corollary*}

\declaretheorem[name=Proposition,numbered=no]{proposition*}

\declaretheorem[numberlike=equation]{claim}
\declaretheorem[name=Claim,numbered=no]{claim*}

\declaretheorem[numberlike=equation]{conjecture}
\declaretheorem[name=Conjecture,numbered=no]{conjecture*}

\declaretheorem[name=Question,numbered=no]{question*}

\declaretheoremstyle[bodyfont=\it,qed=$\lozenge$]{defstyle} 

\declaretheorem[numberlike=equation,style=defstyle]{definition}
\declaretheorem[unnumbered,name=Definition,style=defstyle]{definition*}

\declaretheorem[unnumbered,name=Example,style=defstyle]{example*}

\declaretheorem[unnumbered,name=Notation=defstyle]{notation*}

\declaretheorem[unnumbered,name=Construction,style=defstyle]{construction*}

\declaretheorem[numberlike=equation,style=defstyle]{remark}
\declaretheorem[unnumbered,name=Remark,style=defstyle]{remark*}

\declaretheorem[unnumbered,name=Assumption,style=defstyle]{assumption*}

\usepackage{nth}
\usepackage{intcalc}
\usepackage{etoolbox}
\usepackage{xstring}
\hypersetup{
}

\usepackage{ifpdf}
\ifpdf
\else
\usepackage[quadpoints=false]{hypdvips}
\fi

\newcommand{\ECCCurl}[2]{http://eccc.hpi-web.de/report/\ifnumcomp{#1}{>}{93}{19}{20}#1/#2/}

\newcommand{\shortECCC}[2]{\texttt{\href{http://eccc.hpi-web.de/report/\ifnumcomp{#1}{>}{93}{19}{20}#1/#2/}{eccc:TR#1-#2}}}

\newcommand{\parseECCC}[1]{
\StrSubstitute{#1}{TR}{}[\tmpstring]%
\IfSubStr{\tmpstring}{/}{ 
\StrBefore{\tmpstring}{/}[\ecccyear]%
\StrBehind{\tmpstring}{/}[\ecccreport]%
}{
\StrBefore{\tmpstring}{-}[\ecccyear]%
\StrBehind{\tmpstring}{-}[\ecccreport]%
}%
\shortECCC{\ecccyear}{\ecccreport}}

\usepackage{algorithmicx}
\usepackage{algorithm} 
\usepackage[noend]{algpseudocode}

\algrenewcommand\algorithmicindent{1.0em}%
\usepackage[T1]{fontenc}
\usepackage[inline]{enumitem}

\usepackage{cleveref}
\usepackage{url}

\usepackage[numbers]{natbib}
\usepackage{bibentry}

\newcommand{\ignore}[1]{}
\newcommand{\nobibentry}[1]{{\let\nocite\ignore\bibentry{#1}}}

\newcommand{\bra}[1]{\ensuremath{\left( #1 \right) }}
\newcommand{\cbra}[1]{\ensuremath{\left\{ #1 \right\} }}

\newcommand{\vecalpha}{\boldsymbol{\alpha}}

\newcommand{\eqdef}{\vcentcolon=}

	\renewcommand{\vec}[1]{{\mathbf{#1}}}

	\makeatletter
	\newcommand{\va}{{\vec{a}}\@ifnextchar{^}{\!\:}{}}
	\newcommand{\vb}{{\vec{b}}\@ifnextchar{^}{\!\:}{}}
	\newcommand{\vc}{{\vec{c}}\@ifnextchar{^}{\!\:}{}}
	\newcommand{\vd}{{\vec{d}}\@ifnextchar{^}{\!\:}{}}
	\newcommand{\ve}{{\vec{e}}\@ifnextchar{^}{\!\:}{}}
	\newcommand{\vy}{{\vec{y}}\@ifnextchar{^}{\!\:}{}}
	\newcommand{\vs}{{\vec{s}}\@ifnextchar{^}{\!\:}{}}
	\newcommand{\vt}{{\vec{t}}\@ifnextchar{^}{\!\:}{}}
	\newcommand{\vx}{{\vec{x}}\@ifnextchar{^}{}{}}		
	\newcommand{\vz}{{\vec{z}}\@ifnextchar{^}{\!\:}{}}
	\newcommand{\vv}{{\vec{v}}\@ifnextchar{^}{\!\:}{}}
	\newcommand{\vu}{{\vec{u}}\@ifnextchar{^}{\!\:}{}}
	\newcommand{\vf}{{\vec{f}}\@ifnextchar{^}{\!\:}{}}
	\newcommand{\vg}{{\vec{g}}\@ifnextchar{^}{\!\:}{}}
	\newcommand{\vr}{{\vec{r}}\@ifnextchar{^}{\!\:}{}}
	\newcommand{\vw}{{\vec{w}}\@ifnextchar{^}{\!\:}{}}

	\newcommand{\vY}{{\vec{Y}}\@ifnextchar{^}{\!\:}{}}
	\newcommand{\vX}{{\vec{X}}\@ifnextchar{^}{}{}}		
	\newcommand{\vZ}{{\vec{Z}}\@ifnextchar{^}{\!\:}{}}
	\newcommand{\vG}{{\vec{G}}\@ifnextchar{^}{\!\:}{}}
	
	\makeatother

	\newcommand{\vaa}{{\vecalpha}}

\renewcommand{\C}{\mathbb{C}}
\renewcommand{\N}{\mathbb{N}}

\newcommand{\cA}{{\mathcal{A}}}
\newcommand{\cB}{{\mathcal{B}}}
\newcommand{\cC}{{\mathcal{C}}}
\newcommand{\cD}{{\mathcal{D}}}

\newcommand{\cI}{{\mathcal{I}}}
\newcommand{\cJ}{{\mathcal{J}}}
\newcommand{\cK}{{\mathcal{K}}}
\newcommand{\calL}{{\mathcal{L}}}
\newcommand{\cQ}{{\mathcal{Q}}}
\newcommand{\calP}{{\mathcal{P}}}

\newcommand{\cT}{{\mathcal{T}}}

\newcommand{\cW}{{\mathcal{W}}}
\ifdefined \cS
	\renewcommand{\cS}{{\mathcal{S}}}
\else
	\newcommand{\cS}{{\mathcal{S}}}
\fi
\newcommand{\calS}{{\mathcal{S}}}

\newcommand{\cZ}{{\mathcal{Z}}}

\newcommand{\ideal}[1]{\left \langle{#1}\right \rangle}
\newcommand{\MVar}[2]{{#1}_1,\ldots, {#1}_{#2}}

\newcommand{\PRing}[3]{\mathbb{#1}[ \MVar{#2}{#3}]}

\newcommand{\CRing}[2]{\PRing{C}{#1}{#2}}

\newcommand{\spn}[1]{\operatorname{span}\{{#1}\}}

\newcommand{\Anote}[1]{\textcolor{red}{A: #1}\PackageWarning{amir: }{#1}}

\def\epsilon{\varepsilon} 

\newcommand{\shir}[1]{\newline{\color{blue} Shir Replay: {#1}}}

\newcommand{\MS}{\text{Lin}}

\newcommand{\LRV}{\cC_{[V]}}
\newcommand{\LRIdealV}{\cC_{\ideal{V}}}

\newcommand{\HRV}{\cJ_{[V]}}
\newcommand{\HRIdealV}{\cJ_{\ideal{V}}}

\newcommand\card[1]{\left|{#1}\right|}

\newif\ifEK
\EKtrue

\date{}

\title{Robust Sylvester-Gallai type theorem for quadratic polynomials}
\author{Shir Peleg\thanks{
		The Blavatnik School of Computer Science, Tel Aviv University, Tel Aviv, Israel. Emails: \texttt{shirpele@tauex.tau.ac.il, shpilka@tauex.tau.ac.il}. The research leading to these results has received funding from the  Israel Science Foundation (grant number 514/20) and from the Len Blavatnik and the Blavatnik Family foundation. }  \and  Amir Shpilka\footnotemark[1]}
\begin{document}
\maketitle

\begin{abstract}

In this work we extend the robust version of the Sylvester-Gallai theorem, obtained by Barak, Dvir, Wigderson and Yehudayoff, and by Dvir, Saraf and Wigderson, to the case of quadratic polynomials.
Specifically, we prove that if $\cQ\subset \C[x_1.\ldots,x_n]$ is a finite set, $|\cQ|=m$, of irreducible quadratic polynomials that satisfy the following condition:
\begin{itemize}
	\item There is $\delta>0$ such that for every $Q\in\cQ$ there are at least $\delta m$ polynomials $P\in \cQ$ such that whenever $Q$ and $P$ vanish then so does a third polynomial in $\cQ\setminus\{Q,P\}$,
\end{itemize}
then $\dim(\spn{\cQ})=\poly(1/\delta)$. 

The work of Barak et al. and Dvir et al. studied the case of linear polynomials and proved an upper bound of $O(1/\delta)$ on the dimension (in the first work an upper bound of $O(1/\delta^2)$ was given, which was improved to $O(1/\delta)$ in the second work). 
\end{abstract}

%

\thispagestyle{empty}
\newpage

\tableofcontents

\thispagestyle{empty}
\newpage
\pagenumbering{arabic}

\section{Introduction}\label{sec:intro}

In this paper we prove a robust version of a result of \cite{ShpilkaSG}: Let  $\cT\subset\C[x_1,\ldots,x_n]$ be a finite set of polynomials. We say that $Q_1(\vx),Q_2(\vx)\in \cQ$ satisfy the \emph{Polynomial Sylvester-Gallai condition} (PSG-condition for short) if there is a third polynomial $Q_3(\vx)\in \cQ$ such that $Q_3(\vx)$ vanishes whenever $Q_1(\vx)$ and $Q_2(\vx)$ vanish. We  prove that if $\cT\subset \C[x_1,\ldots,x_n]$ is a finite set containing only irreducible  quadratic polynomials,
such that for every $Q\in \cT$  a $\delta$ fraction of the polynomials in $\cT$ satisfy the PSG-condition with $Q$, 
then $\dim\bra{\spn{\cT}}=\poly(1/\delta)$. 

The motivation for proving this result, besides its own appeal, is two fold: a similar theorem played an important role in the polynomial identity testing (PIT for short) problem for small depth algebraic circuits, one of the fundamental open problems in theoretical computer science, see \cite{Peleg-Shpilka-PIT}; and it is also related to a long line of work extending and generalizing the original Sylvester-Gallai theorem \cite{Melch41,Galai44} such as the results in \cite{DBLP:journals/dcg/Chvatal04,DBLP:journals/dm/PretoriusS09,DBLP:journals/dcg/DvirH16,DBLP:journals/dcg/BasitDSW19}. In particular, our result builds and generalizes a result of \cite{barak2013fractional,DSW12}, that can be viewed as proving an analogous claim for the case of degree-$1$ polynomials.
Such results are useful in discrete geometry \cite{barak2013fractional,DSW12}, in the study of locally correctable codes, for reconstruction of certain depth-$3$ circuits \cite{DBLP:journals/siamcomp/Shpilka09,DBLP:conf/coco/KarninS09,DBLP:conf/coco/Sinha16} and more. See the survey of Dvir on incidence geometry for some applications of Sylvester-Gallai type theorems \cite{DBLP:journals/fttcs/Dvir12}.

We next give background on the Sylvester-Gallai theorem, and some of its variants, and then discuss the connection to the polynomial identity testing problem.

\paragraph{Sylvester-Gallai type theorems:}

The Sylvester-Gallai theorem (SG-theorem) asserts that given a set $S=\cbra{\vv_1,\ldots,\vv_m}\subset\R^n$ such that $S$ is not contained in a line, there must be a line that contains exactly two points from $S$.
It was first conjectured by Sylvester in 1893 \cite{Syl1893} and then
proved, independently, by Melchior in 1941 \cite{Melch41} and Gallai \cite{Galai44} in 1943 (in an answer to the same question posed by Erd\"os, who was unaware of Melchior's result \cite{Erdos43}).
There are many extensions and generalizations of the theorem. We shall state a few that are related to this work. It is also helpful to think of the contra-positive statement. We say that a set of points is a Sylvester-Gallai configuration (SG-configuration for short) if every line that intersects the set at two points, must contain at least three points from the set. Thus, an SG-configuration in $\R^n$ must be colinear.

In \cite{Serre66} Serre, aware that the original formulation of the theorem is not true over $\C$ asked ``Is there a nonplanar version of the Sylvester-Gallai configuration over the field of complex numbers?''  Kelly proved that the answer is no, i.e. that every finite set of points in $\C^n$  satisfying the SG-condition is planar \cite{Kelly86}. Edelstein and Kelly proved a colorful variant of the problem:  
if three finite sets of points in $\R^n$ satisfy that every line passing through points from two different sets also contains a point from the third set, then, the points belong to a three-dimensional affine space. This result can be extended to any constant number of sets. Many more extensions and generalizations of the SG-theorem are known, e.g. \cite{Hansen65,DBLP:journals/dcg/DvirH16}.
The survey by Borwein and Moser \cite{BorwenMoser90} is a good resource on the SG-Theorem and some of the different variants that have been studied in the past.


More recently, Barak et al. \cite{barak2013fractional} and Dvir, Saraf and Wigderson \cite{DSW12}, motivated by questions on locally decodable codes and construction of rigid matrices, proved a \emph{robust} (or fractional) version of the SG-theorem: 	


\begin{definition}[$\delta$-SG configuration]\label{def:delta-SGConf}
	We say that a set of points $v_1,\ldots ,v_m\in \C^n$ is a $\delta$-SG configuration if for every $i\in [m]$ there exists at least $\delta (m-1)$ values of $j \in [m]$ such that the line through $v_i,v_j$ contains a third point in the set.
\end{definition}

\begin{theorem}[Theorem 1.9 of \cite{DSW12}]\label{thm:DSW}
		Let $V = \lbrace v_1,\ldots ,v_m\rbrace \subset \C^n$ be a $\delta$-SG configuration. Then $\dim(\spn{v_1,\ldots ,v_m}) \leq \frac{12}{\delta}+1$.
\end{theorem}


%
%
%

\paragraph{Algebraic generalizations of Sylvester-Gallai type theorems:}

Although the Sylvester-Gallai theorem and
Theorem~\ref{thm:DSW} 
are formulated in the setting of discrete geometry, there is a very natural algebraic formulation:
If a finite set of pairwise linearly independent vectors, $\cS \subset \C^n$, has the property that every two vectors span a third vector in the set, then the dimension of $\calS$ is at most $3$. The proof is immediate from Kelly's theorem: pick a subspace $H$ of codimension $1$, which is in general position with respect to the vectors in $\cS$. The intersection points $p_i = H \cap \spn{s_i}$, for $s_i\in \cS$, satisfy the SG-condition over $\C$. Therefore, $\dim(S) \leq 3$.  An equivalent formulation, in the case of linear functions, is the following: If a finite set of pairwise linearly independent linear forms, $\calL \subset \CRing{x}{n}$, has the property that for every two forms $\ell_i,\ell_j \in \calL$ there is a third  $\ell_k \in \calL$, such that $\ell_k=0$ whenever $\ell_i=\ell_j=0$,  then the linear dimension
of $\calL$ is at most $3$. To see the equivalence note that it must be the case that $\ell_k \in \spn{\ell_i,\ell_j}$ and thus the coefficient vectors of the forms in the set satisfy the condition for the (vector version of the) SG-theorem, and the bound on the dimension follows. Observe that the last example shows that in the case of linear functions the PSG-condition and the SG-condition are equivalent. 
The last formulation can now be generalized to higher degree polynomials. In particular, the following conjecture was raised by Gupta \cite{Gupta14}.


\begin{definition}[PSG-configuration]\label{def:PSG-condition}
\sloppy	Let $\cT\subset\C[x_1,\ldots,x_n]$ be a set of polynomials.
	We say that $Q_1,Q_2\in \cT$ satisfy the \emph{Polynomial Sylvester-Gallai condition} (PSG-condition for short) if there is a third polynomial $Q_3(\vx)\in \cT$ such that $Q_3$ vanishes whenever $Q_1$ and $Q_2$ vanish. 
	
	We say that a set $\cT$ is a PSG-configuration if every two polynomials $Q_1,Q_2\in \cT$ satisfy the PSG-condition.	
\end{definition}

\begin{problem}[Conjecture 2 of \cite{Gupta14}]\label{prob:sg-alg}
	There is a function $\lambda:\N\to \N$ such that  
for any finite set $\cT\subset \CRing{x}{n}$ of pairwise linearly independent and irreducible polynomials, of degree at most $r$, that satisfy the PSG-condition, it holds that the \emph{algebraic rank} of $\cT$ is at most $\lambda(r)$.
\end{problem}

This problem was answered affirmatively, with a stronger conclusion, in the case of quadratic polynomials ($r=2$) in \cite{ShpilkaSG}.

\begin{theorem}[Theorem 1.7 of \cite{ShpilkaSG}]\label{thm:shp-sg}
	There is a constant $\lambda$ such that the following holds for every $n\in\N$. Let $\cT\subset  \CRing{x}{n}$ consist of homogeneous quadratic polynomials, such that each $Q\in\cT$ is either irreducible or a square of a linear function. If $\cT$ satisfies the PSG-condition then  $\dim\bra{\spn{\cT}}\leq \lambda$.
\end{theorem}

Motivated by applications for the polynomial identity testing problem, Gupta \cite{Gupta14} and Beecken, Mittmann and Saxena \cite{DBLP:journals/iandc/BeeckenMS13} also raised the following colorful variant, which generalizes the Edelstein-Kelly theorem.

\begin{conjecture}[Conjecture 30 of \cite{Gupta14}]\label{con:gupta-ek}

	There is a function $\lambda:\N\to \N$ such that the following holds for every $r,n\in\N$.
	Let $R, B, G$ be finite disjoint sets of pairwise linearly independent, irreducible, homogeneous polynomials in $\C[x_1,\ldots,x_n]$ of degree $\leq r$  such that   for every pair $Q_1,Q_2$ from distinct sets there is a $Q_3$
	in the remaining set so that whenever $Q_1$ and $Q_2$ vanish then also $Q_3$ vanishes. Then the algebraic rank of $(R \cup B \cup G)$ is at most $\lambda(r)$.
\end{conjecture}

This problem was also answered affirmatively, with the same stronger conclusion, in \cite{ShpilkaSG}, for the case of quadratic polynomials.

\begin{theorem}[Theorem 1.8 of \cite{ShpilkaSG}]\label{thm:shp-ek}
	There is a constant $\lambda$ such that the following holds for every $n\in\N$.
	Let $\cT_1,\cT_2$ and $\cT_3$ be finite sets of homogeneous quadratic polynomials over $\C$ satisfying the following properties:
	\begin{itemize}
		\item Each $Q\in\cup_i\cT_i$ is either irreducible or a square of a linear function.
		\item No two polynomials are multiples of each other (i.e., every pair is linearly independent).
		\item For every two polynomials $Q_1$ and $Q_2$ from distinct sets there is a polynomial $Q_3$ in the third set so that whenever $Q_1$ and $Q_2$ vanish then also $Q_3$ vanishes. 
	\end{itemize}
	Then $\dim\bra{\spn{\cup_i\cT_i}}$ has dimension $O(1)$.
\end{theorem}


\paragraph{PIT and Sylvester-Gallai type theorems:}

The PIT problem asks to give a deterministic algorithm that given an arithmetic circuit as input determines whether it computes the identically zero polynomial. The circuit can be given either via a description of its graph of computation (white-box model) or via oracle access to the polynomial that it computes (black-box model). This is a fundamental problem in theoretical computer science that has received a lot of attention from researchers in the last two decades. Besides of being a natural and elegant question, the PIT problem is important due  its connections to lower bounds for arithmetic circuits (hardness-randomness tradeoffs) \cite{DBLP:conf/stoc/HeintzS80,DBLP:conf/fsttcs/Agrawal05,DBLP:journals/cc/KabanetsI04,DBLP:journals/siamcomp/DvirSY09,DBLP:conf/coco/ChouKS18}; its relation to other derandomization problems such as finding perfect matching deterministically, in parallel, \cite{DBLP:journals/cacm/FennerGT19,DBLP:conf/focs/SvenssonT17}, derandomizing factoring algorithms \cite{DBLP:journals/cc/KoppartySS15}, derandomization questions in geometric complexity theory \cite{Mulmuley-GCT-V,DBLP:conf/approx/ForbesS13}; its role in algebraic natural proofs \cite{DBLP:journals/toc/ForbesSV18,DBLP:journals/corr/Grochow0SS17}. In particular, PIT appears to be  the most general algebraic derandomization problem. For more on the PIT problem see  \cite{DBLP:journals/fttcs/ShpilkaY10,Saxena09,Saxena14, ForbesThesis}. For a survey of algebraic hardness-randomness tradeoffs see  \cite{DBLP:journals/eatcs/0001S19}.

A beautiful line of work has shown that deterministic algorithms for the PIT problem for homogeneous depth-$4$ circuits or for depth-$3$ circuits would lead to deterministic algorithms for general circuits \cite{DBLP:conf/focs/AgrawalV08,KKS16Jour}. This makes small depth circuits extremely interesting for the PIT problem. This is also the setting where  Sylvester-Gallai type theorems play an important role. 
 The relation between (colored-versions of the) SG-theorem and deterministic PIT algorithms for depth-$3$  circuits was observed in \cite{DBLP:journals/siamcomp/DvirS07}. The work of \cite{DBLP:conf/focs/KayalS09,DBLP:journals/jacm/SaxenaS13} used this relation to obtain polynomial- and quasi-polynomial-time PIT algorithms for depth-$3$ circuits, depending on the characteristic. Currently, the best algorithm for PIT of depth-$3$ circuits was obtained through a different yet highly related approach in \cite{DBLP:journals/siamcomp/SaxenaS12}. As the SG-theorem played such an important role in derandomizing PIT for depth-$3$ circuits, it was asked whether a similar approach could work for depth-$4$ circuits. This motivated  \cite{DBLP:journals/iandc/BeeckenMS13,Gupta14} to raise Problem~\ref{prob:sg-alg} and Conjecture~\ref{con:gupta-ek}. In \cite{Peleg-Shpilka-PIT} we gave a positive answer to  Conjecture~\ref{con:gupta-ek} for the case of degree-$2$ polynomials ($r=2$). Interestingly, \autoref{thm:DSW} played a crucial role in the proof, as well as in the proofs of \cite{ShpilkaSG,Peleg-Shpilka-SG}. Studying the proofs of \cite{ShpilkaSG,Peleg-Shpilka-SG,Peleg-Shpilka-PIT} leads to the conclusion that in order to solve Problem~\ref{prob:sg-alg} and Conjecture~\ref{con:gupta-ek} for degrees larger than $2$, we must first obtain a result analogous to \autoref{thm:DSW}.

\paragraph{Our results:}

In this work we are able to prove an analog of \autoref{thm:DSW} for quadratic polynomials. We hope that this result will lead to an extension of the works 
\cite{ShpilkaSG,Peleg-Shpilka-SG,Peleg-Shpilka-PIT} to higher degree polynomials.

\begin{definition}[$\delta$-PSG-configuration]\label{def:delta-PSG-condition}
	Let $\cQ\subset\C[x_1,\ldots,x_n]$ be a set of polynomials.
	We say that a finite set of polynomials $\cQ$ is a $\delta$-PSG configuration if for every $Q\in\cQ$ there are at least $\delta\cdot |\cQ|$ polynomials $P\in\cQ$ such that $Q$ and $P$ satisfy the PSG condition.
\end{definition}

\begin{theorem}\label{thm:psg-robust}
	Let $\cQ\subset\C[x_1,\ldots,x_n]$ a finite set of irreducible quadratic polynomials. If $\cQ$ is  a $\delta$-PSG configuration then $\dim(\spn{\cQ})=O(1/\delta^{16})$.
\end{theorem}
%

\begin{remark}
	The same conclusion holds even if we allow irreducible polynomials of degree
	 at most $2$ (i.e. if we allow linear functions). The proof is similar in nature, with more case
	analysis, and so we decided to omit it.
\end{remark}

Note that this is robust version of \autoref{thm:shp-sg} in the same sense that \autoref{thm:DSW} is a robust version of the SG-theorem. 

\begin{remark}
	While the result in \autoref{thm:DSW} tight (up to the constant in the big Oh), we do not believe that the result of \autoref{thm:psg-robust} is tight. In particular, we believe that the  upper bound should be  $O(1/\delta)$.
\end{remark}
\begin{remark}
After this work was completed, we learned from Garg, Oliviera and Sengupta that they have independently obtained the same result and are in the process of writing it \cite{GOS}. Both proofs heavily rely on case analysis and are similar in nature with small technical differences. The main difference between the proofs is that we relied on Theorem 1.10 of \cite{ShpilkaSG} (see \autoref{thm:structure}) while they proved a generalization of  that theorem.
\end{remark}

\subsection{Proof idea}\label{sec:proof-idea}

To explain the proof we will use some notation that we define in Section~\ref{sec:notation}. For example, $\ideal{\cdot}$ denotes an ideal, $\sqrt{\ideal{\cdot}}$ denotes the radical of the ideal, and $\C[V]_2$ denotes the space of all quadratic polynomials defined only using the linear forms in $V$.

At the heart of all previous work lies an algebraic theorem, classifying the cases in which a quadratic polynomial vanishes when two other quadratics vanish (actually, for \cite{Peleg-Shpilka-SG,Peleg-Shpilka-PIT} a more general result was needed - a  characterization of the different cases in which a product of quadratic polynomials vanishes whenever two other quadratics vanish).

\begin{theorem}[Theorem 1.10 of \cite{ShpilkaSG}]\label{thm:structure}
	Let $A,B$ and $C$ be $n$-variate, homogeneous, quadratic polynomials, over $\C$, such  that whenever $A$ and $B$ vanish then so does $C$. Then, one of the following cases must hold:
	\begin{enumerate}[label={(\roman*)}]
		\item \label{case:span} 
		$C$ is in the linear span of $A$ and $B$.  
		\item \label{case:rk1}
		There exists a non trivial linear combination of the form $\alpha A+\beta B = \ell^2$ for some linear form $\ell$. 
		\item \label{case:2} 
		There exist two linear forms $\ell_1$ and $\ell_2$ such that when setting $\ell_1=\ell_2=0$ we get that $A$ and $B$ (and consequently $C$) vanish.  
	\end{enumerate}
\end{theorem}

\sloppy
The high level idea in the proof of \autoref{thm:shp-sg} (which was generalized in \cite{Peleg-Shpilka-SG,Peleg-Shpilka-PIT}), includes two steps; The first step  constructs a linear space of linear forms $V$, and a subset $\cJ\subset Q$, both of constant dimension such that a vast majority of the polynomials in $\cQ$ are in $\spn{\cJ,\cQ\cap\ideal{V}}$.\footnote{ \cite{ShpilkaSG} had different notations, and $|\cJ| = 1$.} Implementing this idea requires a lot of  case analysis, according to \autoref{thm:structure}. In the second step the dimension of $\cQ\cap \ideal{V}$ is upper bounded.

The idea outlined above  heavily relies on the fact that when $\delta=1$, the set $\cQ\cap\ideal{V}$ is a PSG-configuration in itself. Indeed, let $Q_1, Q_2\in \cQ\cap\ideal{V}$. When $\delta =1$ it follows that there is $Q_3\in \cQ$ such that $Q_3\in \sqrt{\ideal{Q_1,Q_2}}\subseteq \ideal{V}$.
In order to bound  the dimension of $\cQ\cap \ideal{V}$, \cite{ShpilkaSG} "projected" $V$ to a one dimensional space $\spn{z}$ (where $z$ is a new variable). Every polynomial $Q_i\in \cQ\cap \ideal{V}$ is mapped to a polynomial of the form $z\cdot \ell_i$, for some linear form $\ell_i$. Then, it is proved that the $\ell_i$'s form an SG-condition.\footnote{The reader should take note that this is  a very high-level simplification of one part in the proof. For more details see the "easy-case" in \cite{Peleg-Shpilka-SG,Peleg-Shpilka-PIT}.}

This  technique fails when $\delta\in(0,1)$. First, we cannot expect to prove that $\cQ\cap \ideal{V}$ is a $\delta'$-PSG configuration by itself (even when we allow smaller, yet fixed, $\delta'\leq \delta$). For example, since $\delta<1$, it may be the case that (many polynomials) $Q\in \cQ\cap \ideal{V}$ have \emph{all} of their neighbors outside  $\cQ\cap \ideal{V}$. 
Furthermore, even if we knew that $\cQ\cap \ideal{V}$ is a $\delta'$-PSG configuration, then it is not clear that by  following the lines of \cite{ShpilkaSG} and mapping $\ideal{V}$ to $\spn{z}$,  the resulting $\ell_i$s, form a $\delta'$-PSG configurations. The reason for that is a bit subtle: note that it may be the case that many polynomials $Q\in \cQ\cap \ideal{V}$ were mapped to  $\spn{z^2}$. Thus, it may be the case that all the neighbors of some $z\cdot \ell$ are in $\spn{z^2}$, which gives us no information at all about $\ell$. In contrast, in \cite{ShpilkaSG}, since $\delta=1$, we could get information about $\ell$ by its interaction with polynomials not in $\spn{z^2}$.


In order to overcome these issues, we needed to develop new techniques, and improve the characterization given in  \autoref{thm:structure}\ref{case:2} (see \autoref{cor:case-3-strong}). Next, we present the outline of the proof in more details.

 We start with the same line of  constructing a linear space of linear forms $V$, and a subset $\cJ\subset Q$, both of dimension $O(\mathsf{poly}(\frac{1}{\delta}))$ such that $\cQ\subseteq \spn{\cJ, \ideal{V}}$. We partition $\cQ$ to four sets: $\LRV=\cQ\cap \C[V]_2$;   $\LRIdealV=(\cQ\cap \ideal{V}) \setminus  \LRV$; $\HRV=\cQ  \cap \spn{\cJ \cup C[V]_2}$; and the remaining set $\HRIdealV=\cQ \cap \spn{\cJ \cup\ideal{V}}\setminus \HRV$. We already know that $\dim(\LRV\cup\HRV)$ is small, so we only have to bound the dimension of $\LRIdealV\cup \HRIdealV$. 

Let us focus on $\LRIdealV$. We would like to prove that we can add a few linear functions to $V$ to get a subspace $U$ such that  $\LRIdealV\subset\C[U]_2$. Let $P\in\LRIdealV$.
First we  consider the case that many of $P$'s neighbors (i.e. those polynomials with which $P$ satisfies the PSG-condition) are in $\LRV\cup\LRIdealV$. To handle this case we   strengthen  \autoref{thm:structure}\ref{case:2} and  use it to show that if $Q\in \LRV$ is a neighbor of $P$ then the polynomial $Q'\in\sqrt{\ideal{P,Q}}$ is unique (see \autoref{cor:unique-T}). This means that by moving the linear functions on which $P$ depends to $U$, we move many polynomials from $\LRIdealV$ to $\C[V+U]_2$. 

Next we consider the case where $P$ has ``many'' neighbors in $\HRV\cup \HRIdealV$. To handle this case we first prove that $P$ can only satisfy \cref{thm:structure}\ref{case:span} with polynomials in $\HRV\cup \HRIdealV$. We prove that under this condition, there is a ``large'' subset of $\LRIdealV$ that is of constant dimension. Thus, by adding a few linear functions to $U$, we move many polynomials from $\LRIdealV$ to  $\C[V+U]_2$ (see \autoref{cla:decreaseCideal}). We can continue this process as long as $\LRIdealV$ is large enough, as the amount of polynomials that we move at any step depends on $|\LRIdealV|$. Therefore, when this process terminates we still have to deal with a set $\LRIdealV$ that is not large but not too small either (it is of size $\Omega(\delta m)$). Now, we turn our attention to ${\HRIdealV}$. Using similar arguments, and relying on the fact that $\card{\LRIdealV}$ is small, we prove that we can add a few linear functions to $U$ and make $\card{\cJ_{\ideal{V+U}}}$ small. Having achieved that, we prove that if both $\card{\cC_{\ideal{V+U}}}$ and $\card{\cJ_{\ideal{V+U}}}$ are small then they are in fact, empty (see \autoref{cla:C,J-large}).

\subsection{Discussion}\label{sec:discuss}

There are two distinct goals to the line of work \cite{ShpilkaSG,Peleg-Shpilka-SG,Peleg-Shpilka-PIT}, including this paper. The first is obtaining higher degree geometric extensions of the Sylvester-Gallai and Edelstein-Kelley theorems.

 From the complexity theoretic point of view, the goal is to eventually obtain PIT algorithms for   $\Sigma^{[k]}\Pi^{[d]}\Sigma\Pi^{[r]}$ circuits, for any $k,r=O(1)$. 
Currently we  have a polynomial time PIT algorithm only for the case  $k=3$ and $r=2$ \cite{Peleg-Shpilka-PIT}.
%
To understand such a difficult question one has to start somewhere, and  the case $k=3$ and $r=2$ was a natural starting point for the investigation (especially as no subexponential time PIT algorithm, even for $\Sigma^{[3]}\Pi^{[d]}\Sigma\Pi^{[2]}$ circuits, was known prior to  \cite{Peleg-Shpilka-PIT}). Since so little is known, we believe that a natural approach for advancing is to first extend the results of \cite{Peleg-Shpilka-PIT} to higher degrees (i.e. higher values of $r$), and then for a higher top fan-in (i.e. higher values of $k$). 
 Before we explain the difficulties in going to higher degrees we recall that \cite{Peleg-Shpilka-PIT} needed the following strengthening of  \autoref{thm:shp-ek} for their  PIT algorithm.

\begin{theorem}[Theorem 1.6 in \cite{Peleg-Shpilka-PIT}]\label{thm:PS-PIT}
	There exists a universal constant $\lambda$ such that the following holds. 
	Let $\cT_1,\cT_2, \cT_3\subset\C[x_1,\ldots,x_n]$ be finite sets of pairwise linearly independent homogeneous polynomials satisfying the following properties:
	\begin{itemize}
		\item Each $Q\in\cup_{j\in[3]}\cT_j$ is either irreducible quadratic  or a square of a linear function.
		\item Every two polynomials $Q_1$ and $Q_2$ from distinct sets satisfy that whenever they vanish then the product of all the polynomials in the third set  vanishes as well. 
	\end{itemize}
	Then,  $\dim(\spn{\cup_{j\in[3]}\cT_j})\leq \lambda$.
\end{theorem}

There are several difficult hurdles in going from $r=2$ to general $r$, or even to $r=3$, if we wish to continue working in the framework of \cite{ShpilkaSG,Peleg-Shpilka-SG,Peleg-Shpilka-PIT} (and this paper). The first is understanding what is the correct generalization of  \autoref{thm:structure} to higher degrees, as this theorem lies at the heart of all these papers. 
A second hurdle is obtaining a robust version of \autoref{thm:PS-PIT}. First for $r=2$ and then for higher degrees.  

For extending \autoref{thm:structure} to higher degrees it seems natural to find an extension to $r=3$. While it seems that such an approach could last forever and lead nowhere (as we will then have to prove a result for $r=4$ etc.), we believe that understanding the case $r=3$  can shed more light on the general case, as sometimes going from degree $2$ to $3$ is as difficult as the general case.

Once we prove such a structural theorem, we will need to extend \autoref{thm:PS-PIT} to higher values of $r$. An important tool in the proof of \autoref{thm:PS-PIT} was a robust version of the EK-theorem. 
\begin{definition}[$\delta$-EK configuration]\label{def:partial-EK}
		We say that the sets $\cT_1,\cT_2,\cT_3\subset \C^n$ form a $\delta$-EK configuration if for every $i \in [3]$ and $p\in \cT_i$ a
		$\delta$ fraction of the vectors $q$ in the union of the two other sets satisfy that  $p$ and $q$ span some vector in the third set (the one not containing $p$ and $q$). 
		We refer to a $1$-EK configuration as simply an EK-configuration.
	\end{definition}
\begin{theorem}[Theorem 3.9 of \cite{Peleg-Shpilka-PIT}]\label{thm:partial-EK-robust}
	Let $0<\delta \leq 1$ be any constant. Let $\cT_1,\cT_2,\cT_3\subset\C^n$ be disjoint finite sets that form a $\delta$-EK configuration.  Then, $\dim(\spn{\cup_i \cT_i}) = O(1/\delta^3)$.
\end{theorem}
Thus, a natural continuation would be to prove a robust version of \autoref{thm:partial-EK-robust} for quadratic polynomials (i.e. a robust version of \autoref{thm:shp-ek}) and then to extend it to a robust version of \autoref{thm:PS-PIT} and to higher degrees.
While in this paper we only prove a robust version of \autoref{thm:shp-sg},  we believe that with some more technical work this can be extended to a robust version of \autoref{thm:shp-ek} as well. Hence, the next immediate challenge would be to obtain a robust version of \autoref{thm:PS-PIT} (or even of the main result of \cite{Peleg-Shpilka-SG}).
If we obtain such an extension and, in addition extend \autoref{thm:structure} to higher values of $r$, then we expect that a PIT algorithm for the case $k=3$ and $r=3$ would follow. 
More importantly, we believe that this will let us gain important understanding on how to generalize the results for arbitrary values of $r$.

%
%
%
%

\ifEK

\section{Preliminaries}\label{sec:prelim}

In this section we explain our notation and present some basic algebraic preliminaries.

\subsection{Notation and basic facts}\label{sec:notation}

Greek letters $\alpha, \beta,\ldots$ denote scalars from $\C$.  Non-capitalized letters $a,b,c,\ldots$  denote linear forms and $x,y,z$ denote variables (which are also linear forms). Bold faced letters denote vectors, e.g. $\vx=(x_1,\ldots,x_n)$ denotes a vector of variables, $\vaa=(\alpha_1,\ldots,\alpha_n)$ is a vector of scalars, and $\vec{0} = (0,\ldots,0)$ is the zero vector. We sometimes do not use a boldface notation for a point in a vector space if we do not use its structure as a vector. Capital letters such as $A,P,Q$  denote quadratic polynomials whereas $V,U,W$ denote linear spaces, and $I,J$ denote ideals. Calligraphic letters $\cal I,J,F,Q,T$  denote sets. For a positive integer $n$ we denote $[n]=\{1,2,\ldots,n\}$. 


We denote with $\CRing{x}{n}$ the ring of $n$-variate polynomials over $\C$. We shall denote with $\CRing{x}{n}_d$ the linear space of homogeneous polynomials of degree $d$. In particular, $\CRing{x}{n}_1$ is the linear space of all linear forms.
An \emph{Ideal} $I\subseteq \CRing{x}{n}$ is an abelian subgroup that is closed under multiplication by ring elements. For $\calS \subset \CRing{x}{n}$, we  denote with $\ideal{\calS}$, the ideal generated by $\calS$, that is, the smallest ideal that contains $\calS$. For example, for two polynomials $Q_1$ and $Q_2$, the ideal $\ideal{Q_1,Q_2}$ is the set $\CRing{x}{n}Q_1 + \CRing{x}{n}Q_2$. For a linear subspace $V\subseteq \CRing{x}{n}$, we have that $\ideal{V}$ is the ideal generated by any basis of $V$. The \emph{radical} of an ideal $I$, denoted $\sqrt{I}$, is the set of all ring elements, $r$, satisfying that for some natural number $m$ (that may depend on $r$), $r^m \in I$. Hilbert's Nullstellensatz  implies that, in $\C[x_1,\ldots,x_n]$, if a polynomial $Q$ vanishes whenever $Q_1$ and $Q_2$ vanish, then $Q\in \sqrt{\ideal{Q_1,Q_2}}$ (see e.g. \cite{CLO}). We shall often use the notation $Q\in \sqrt{\ideal{Q_1,Q_2}}$ to denote this vanishing condition. For an ideal $I\subseteq \CRing{x}{n}$ we denote by $\CRing{x}{n} /I$ the \emph{quotient ring}, that is, the ring whose elements are the cosets of $I$ in $\CRing{x}{n}$ with the proper multiplication and addition operations. For an ideal $I\subseteq \CRing{x}{n}$ we denote the set of all common zeros of elements of $I$ by $\cZ(I)$.
An ideal $I$ is called \emph{prime} if for every $f$ and $g$ such that $fg\in I$ it holds that either $f\in I$ or $g\in I$.  We next present basic facts about prime ideals that are used throughout the proof.
\begin{fact}
	\begin{enumerate}{}
		\item If $F$ is an irreducible polynomial then $\ideal{F}$ is a prime ideal.
		\item For linear forms $\MVar{a}{k}$ the ideal $\ideal{\MVar{a}{k}} = \ideal{\spn{\MVar{a}{k}}}$ is prime.
		\item If $I$ is a prime ideal then $\sqrt{I}=I$.
	\end{enumerate}
\end{fact}

For $V_1,\ldots,V_k$ linear spaces, we use $\sum_{i=1}^k V_i$ to denote the linear space $V_1 + \ldots + V_k$. For two nonzero polynomials $A$ and $B$ we denote $A\sim B$ if there is a nonzero $\alpha\in\C$ such that $B=\alpha\cdot A$. For a space of linear forms $V = \spn{\MVar{v}{\Delta}}\subseteq \CRing{x}{n}$, we say that a polynomial $P \in \CRing{x}{n}$ depends only on $V$  if the value of $P$ is determined by the values of the linear forms $v_1,\ldots,v_\Delta$. More formally, we say that $P$ depends only on $V$ if there is a $\Delta$-variate polynomial $\tilde{P}$ such that $P \equiv \tilde{P}(v_1,\ldots,v_\Delta)$. We denote by  $\C[V]\subseteq \CRing{x}{n}$ the subring of polynomials that depend only on $V$. Similarly we denote by $\C[V]_2\subseteq \CRing{x}{n}$, the linear subspace of all \textbf{homogeneous} quadratic polynomials that depend only on $V$. 

Another notation that we will use throughout the proof is congruence modulo linear forms.
\begin{definition}\label{def:mod-form}
	Let $V\subset \CRing{x}{n}_1$ be a linear space\footnote{I.e., $V$ is a linear space of linear forms.}  and $P,Q\in \CRing{x}{n}$. We say that $P\equiv_V Q$ if $P-Q \in \ideal{V}$. 
\end{definition}

We end with a simple observation that follows immediately from the fact that the quotient ring $\CRing{x}{n}/{\langle V\rangle}$is a unique factorization domain.
\begin{observation}\label{fact:ufd}
	Let $V\subset \CRing{x}{n}_1$ be a linear space and $P = \prod_{k=1}^t P_k$, $Q = \prod_{k=1}^t Q_k \in \CRing{x}{n}$. If for every $k$, $P_k$ and $Q_k$ are irreducible in $\CRing{x}{n}/{\langle V\rangle}$, and $P \equiv_V Q \not\equiv_V 0$ then, up to a permutation of the indices, $P_k\equiv_V Q_k$ for all $k\in [t]$.
\end{observation}

We shall use this observation implicitly when factoring polynomials modulo  linear spaces of linear forms.

\subsection{Rank of quadratic polynomials}\label{sec:rank}

We next give some facts regarding quadratic polynomials. Many of these facts already appeared in  \cite{Peleg-Shpilka-SG}.

\begin{definition}\label{def:rank-s}
	For a homogeneous quadratic polynomial  $Q$ we denote with $\rank_s(Q)$\footnote{In some recent works  this was defined as $\text{algebraic-rank}(Q)$ or $\text{tensor-rank}(Q)$, but as these notions have different meanings we decided to continue with the notation of \cite{Peleg-Shpilka-SG}.} the minimal $r$ such that there are $2r$  linear forms $\{a_k\}_{k=1}^{2r}$ satisfying $Q=\sum_{k=1}^r a_{2k}\cdot a_{2k-1}$. We call such representation a \emph{minimal representation} of $Q$.
\end{definition}

This is a slightly different notion than the usual one for the rank of a quadratic form,\footnote{The usual definition says that $\rank(Q)$ is the minimal $t$ such that there are $t$ linear forms $\{a_k\}_{k=1}^{t}$, satisfying $Q=\sum_{k=1}^t a_k^2$.}
but it is more suitable for our needs. We note that a quadratic $Q$ is irreducible if and only if $\rank_s(Q)>1$. The next claim shows that a minimal representation is unique in the sense that the
space spanned by the linear forms in it is unique.

\begin{claim}[Claim 2.13 in \cite{Peleg-Shpilka-SG}]\label{cla:irr-quad-span}
	Let $Q$ be a homogeneous quadratic polynomial, and let $Q=\sum_{i=1}^{r}a_{2i-1}\cdot a_{2i}$ and $Q = \sum_{i=1}^{r}b_{2i-1}\cdot b_{2i}$ be two different minimal representations of $Q$. Then, $\spn{\MVar{a}{2r}} =\spn{\MVar{b}{2r}}$.
\end{claim}

\autoref{cla:irr-quad-span} shows that the minimal space is well defined, which allows us to define the notion of a \textit{minimal space} of a quadratic polynomial $Q$.
\begin{definition}\label{def:MS}
	Let Q be a quadratic polynomial. Assume that $\rank_s(Q) = r$, and let $Q = \sum \limits_{i=1}^r a_{2i-1}a_{2i}$ be some minimal representation of $Q$.
	We denote $\MS(Q)\eqdef \spn{\MVar{a}{2r}}$. 
	For a set $\{Q_i\}_{i=1}^{k}$ of quadratic polynomials we denote $\MS(\MVar{Q}{k}) = \sum \limits_{i=1}^k \MS(Q_i)$.
\end{definition}

The following fact is easy to verify.

\begin{fact}\label{cor:containMS}
	Let $Q=\sum_{i=1}^{m}a_{2i-1}\cdot a_{2i}$ be a homogeneous quadratic polynomial, then $\MS(Q)\subseteq \spn{\MVar{a}{2m}}$.
\end{fact}



\begin{claim}[Claim 2.16 in \cite{Peleg-Shpilka-SG}]\label{cla:rank-mod-space}
	Let $Q$ be a homogeneous quadratic polynomial with $\rank_s(Q)=r$, and let $V \subset \CRing{x}{n}$ be a linear space of linear forms such that $\dim(V)=\Delta$. Then, $\rank_s(Q|_{V=0})\geq r-\Delta$.
\end{claim}

\begin{claim}[Claim 2.17 in \cite{Peleg-Shpilka-SG}]\label{cla:ind-rank}
	Let $P_1\in \CRing{x}{k}_2$, and $P_2 = y_1y_2\in \CRing{y}{2}$. Then $ \rank_s (P_1+P_2) = \rank_s(P_1) + 1$. Moreover, $y_1,y_2 \in \MS(P_1+P_2).$
\end{claim}


\begin{corollary}[Corollary 2.18 in \cite{Peleg-Shpilka-SG}]\label{cla:intersection}
	Let $a$ and $b$ be linearly independent linear forms. If $c,d,e$ and $f$ are linear forms such that $ab+cd=ef$ then it must hold that $\dim(\spn{a,b}\cap \spn{c,d})\geq 1$.
\end{corollary}

\begin{claim}[Claim 2.19 in \cite{Peleg-Shpilka-SG}]\label{cla:rank-2-in-V}
	Let $a,b,c$ and $d$ be linear forms, and $V$ be a linear space of linear forms. Assume $\{\vec{0}\} \neq \MS(ab-cd) \subseteq V$. Then, $\spn{a,b} \cap V\neq \{\vec{0}\}$. 
\end{claim}

\begin{definition}\label{def:proj}
	Let $a$ be a linear form and $V\subseteq \CRing{x}{n}_1$ a linear space. We denote by ${V^\perp}(a)$ the projection of $a$ to $V^\perp$ (e.g., by identifying each linear form with its vector of coefficients and taking the usual projection). We also extend this definition to linear spaces: $V^{\perp}(\spn{\MVar{a}{k}}) = \spn{{V^\perp}(a_1),\ldots, {V^\perp}(a_k)}$.
\end{definition}

\begin{claim}\label{cla:lin-rank-r-U}
	Let $Q,Q'$ be quadratic polynomials, and $U$\ a linear space of linear forms. Let $r \in \N$ be a constant. Then, there exists a linear space of linear forms, $V$, of dimension at most $8r$, with the following property: For every $P\in \C[U]_2$ and every linear combination $\alpha Q+\beta Q' +P$ that satisfy $\rank_s(\alpha Q+\beta Q' +P) \leq r$, it holds that $\MS(\alpha Q+\beta Q'+P)\subseteq V+U$.
\end{claim} 
\begin{proof}
	If there are $T,T' \in \C[U]_2$ such that $\rank_s(Q-T),\rank_s(Q'-T') \leq 2r$ then let $V = \MS(Q-T) + \MS(Q'-T')$ and the statement clearly holds.
	Thus, assume without loss of generality that for every $T \in \C[U]_2$, $\rank_s(Q-T) > 2r$.
	Let $A_1 = \alpha_1 Q+ \beta_1 Q' + P_1$ satisfy $\rank_s(A_1) \leq r$, for some $P_1\in \C[U]_2$. Set $V = \MS(A_1)$. If $V$ does not satisfy the statement then let $P_2\in \C[U]_2$ and $A_2 = \alpha_2 Q+ \beta_2 Q' + P_2$ be such that $\rank_s(A_2) \leq r$ and $\MS(A_2) \not \subseteq V+U$. In particular, the vectors  $(\alpha_1, \beta_1)$ and $(\alpha_2, \beta_2)$ are linearly independent. Hence, $Q\in \spn{A_1,A_2, P_1,P_2}$. Consequently, there is $T\in \spn{P_1,P_2}\subseteq \C[U]_2$ such that $\rank_s(Q-T) \leq 2r$, in contradiction.
\end{proof}

\begin{claim}\label{cla:lin-rank-r}
	Let $Q,Q'$ be quadratic polynomials and let $r \in \N$ be a constant. Then, there exists a linear space of linear forms, $V$, of dimension at most $8r$, such that for every linear combination satisfying $\rank_s(\alpha Q+\beta Q') \leq r$ it holds that $\MS(\alpha Q+\beta Q')\subseteq V$
\end{claim} 

\begin{proof}
	This claim follows immediately from \autoref{cla:lin-rank-r-U} with $U = \{\vec{0}\}.$
\end{proof}

\begin{claim}\label{cla:case-rk1-gen}
	Let $P$ be a homogeneous irreducible quadratic polynomial and let $a$ and $b$ be linear forms. Assume that $\prod_{i\in\cI} T_i \in \sqrt{\ideal{P,ab}}$, for some set $\cI$. Then, either $\rank_s(P) = 2$ and $a \in \MS(P)$ or there is $i\in \cI$ such that $T_i = \alpha P + ac$ for some linear form $c$ and  $\alpha \in \C$. 
\end{claim}
\begin{proof}
Since $P$ is irreducible we have  that $\rank_s(P) \geq 2$. We consider two cases concerning $P|_{a=0}$. If $P$ becomes reducible when setting $a = 0$ then  $\rank_s(P|_{a=0})=1$. Therefore it must hold that $\rank_s(P) = 2$ and  $a\in \MS(P)$.
	If $P$ remains irreducible after setting $a = 0$ then   
	$\ideal{P|_{a=0}}$ is a prime ideal. Hence, 
	$\sqrt{\ideal{P|_{a=0}}}=\ideal{P|_{a=0}}$ and thus there is $i\in \cI$ with $T_i|_{a=0}\in \ideal{P|_{a=0}}$. In particular,  $T_i =\alpha P + ac$ for some linear form $c$.
\end{proof}


\ifEK

In \cite{Peleg-Shpilka-SG} the following claim was proved.

\begin{claim}[Claim 2.20 in \cite{Peleg-Shpilka-SG}]\label{cla:linear-spaces-intersaction }
	Let $V = \sum_{i=1}^m V_i$ where $V_i$ are linear subspaces, and for every $i$, $\dim(V_i) = 2$. If for every $i\neq j \in [m]$, $\dim(V_i\cap V_j) = 1$, then either $\dim(\bigcap_{i=1}^m V_i) = 1$ or $\dim(V)=3$. 
\end{claim}

\subsection{Projection mapping}\label{sec:z-map}

This section collects some facts from \cite{Peleg-Shpilka-SG} concerning projections of linear spaces and the effect on relevant quadratic polynomials.


\begin{definition}[Definition 2.21 of \cite{Peleg-Shpilka-SG}]\label{def:z-mapping}
	Let $V = \spn{\MVar{v}{\Delta}}\subseteq \spn{x_1,\ldots,x_n}$ be a $\Delta$-dimensional linear space of linear forms, and let $\{\MVar{u}{{n-\Delta}}\}$ be a basis for $V^\perp$. For $\vaa = (\MVar{\alpha}{\Delta})\in \C^{\Delta}$ we define $T_{\vaa, V} : \CRing{x}{n} \mapsto \C[\MVar{x}{n},z]$, where $z$ is a new variable, to be the linear map given by the following action on the basis vectors: $T_{\vaa, V}(v_i) = \alpha_i z$ and $T_{\vaa, V}(u_i)=u_i$.
\end{definition}

Thus, if we pick $\vaa$ at random, then $T_{\vaa,V}$ projects $V$ to $\spn{z}$ in a random way while keeping the perpendicular space intact. 
Clearly $T_{\vaa, V}$ is a linear transformation, and it defines a ring homomorphism from $\CRing{x}{n}$ to $\C[\MVar{x}{n},z]$ in the natural way. 

\begin{claim}[Claim 2.23 of \cite{Peleg-Shpilka-SG}]\label{cla:res-z-ampping}
	Let $V\subseteq \spn{x_1,\ldots,x_n}$ be a $\Delta$-dimensional linear space of linear forms. Let $F$ and  $G$ be two polynomials that share no common irreducible factor. Then, with probability $1$ over the choice of $\vaa \in [0,1]^{\Delta}$ (say according to the uniform distribution), every common factor of $T_{\vaa, V}(F)$ and $T_{\vaa, V}(G)$ (if there is such) must be in $\C[z]$.
\end{claim}


\begin{claim}[Corollary 2.24 of \cite{Peleg-Shpilka-SG}]\label{cla:still-indep}
	Let $V$ be a $\Delta$-dimensional linear space of linear forms. Let $F$ and  $G$ be two linearly independent, irreducible quadratics, such that $\MS(F),\MS(G)\not\subseteq V$. Then, with probability $1$ over the choice of $\vaa \in [0,1]^{\Delta}$ (say according to the uniform distribution), $T_{\vaa, V}(F)$ and $T_{\vaa, V}(G)$ are linearly independent.
\end{claim}


\begin{claim}[Claim 2.25 of \cite{Peleg-Shpilka-SG}]\label{cla:z-map-rank}
	Let $Q$ be an irreducible quadratic polynomial, and $V$  a $\Delta$-dimensional linear space.
	Then for every $\vaa \in \C^{\Delta}$,  $\rank_s(T_{\vaa, V}(Q)) \geq \rank_s(Q)-\Delta$.
\end{claim}

\begin{claim}[Claim 2.26 of \cite{Peleg-Shpilka-SG}]\label{cla:z-map-dimension}
	\sloppy 
	Let $\cQ$ be a set of quadratics, and $V$ be a $\Delta$-dimensional linear space. Then, if there are linearly independent vectors, $\{\vaa^1,\dots, \vaa^{\Delta}  \}\subset \C^{\Delta}$ such that for every $i$,\footnote{Recall that $\MS(T_{\vaa^i,V}(\cQ))$ is the space spanned by $\cup_{Q\in\cQ}\MS(T_{\vaa^i,V}(\cQ))$.} $\dim(\MS(T_{\vaa^i,V}(\cQ)))\leq \sigma$ then $\dim(\MS(\cQ))\leq (\sigma+1) \Delta$.
\end{claim}

\section{Robust-SG theorems in $\C^n$}\label{sec:robust-C}
We shall need the following  generalizations of \autoref{thm:DSW}. We postpone the proofs to \autoref{sec:appendix}.

\begin{theorem}\label{thm:lin-rob-space}
	Let $0 < \delta \leq 1$ be any constant. Let $W \subset \C^n$
	be an $r$-dimensional space. Let $\mathcal{W} \subset W$ and 
	  $\cK\subset  \C^n \setminus W$ be finite subsets such that no two vectors in $\cT= \cK \cup \cW$ are linearly dependent.
	Assume further that all the elements in $\cK$ satisfy the following relaxed EK-property: For every $p \in \cK$, for at
	least $\delta$ fraction of the points $q \in \cT$
	the span of $p$ and $q$ contains a point in $\cT\setminus \{p,q\}$.
	 Then,	$\dim(\spn{\cT}) \leq O(r + \frac{1}{\delta})$. 
\end{theorem}

We also use the following bi-partitive version of \cite[Corollary 1.11]{DSW12} this is a slight variation of the formulation presented in their paper. 
\begin{claim}\label{cla:case-1-cut}
	
	Let $V = \MVar{v}{n} \subset \C^d$	be a set of $n$ distinct points. Suppose that there is $\cB\subseteq V$ such that there are at least $\delta n^2$  pairs in $\cB\times( V\setminus \cB)$ that lie on a special line. Then there exists a subset $\cB'\subseteq \cB$ such that $\card{\cB'} \geq (\delta/6)n$ and $\text{affine-dim}(\cB') \leq  O(1/\delta)$.

\end{claim}

The important difference between \autoref{cla:case-1-cut} and  \cite[Corollary 1.11]{DSW12} is that \autoref{cla:case-1-cut} guarantees the existence of a low-dimensional subspace that  contains a constant fraction of the points in $\cB$, whereas from \cite[Corollary 1.11]{DSW12} we do not get any guarantee about the fraction of points from $\cB$ in the low-dimensional space. 

\section{Strengthening Case~\ref{case:2} of Theorem~\ref{thm:structure} }

The following claim strengthens \autoref{thm:structure}\ref{case:2} by providing more information on the polynomial in the radical.

\begin{claim}\label{cla:case3-strong}
	Let $P,Q$ and $T$ be irreducible homogeneous quadratic polynomials, such that $T\in \sqrt{\ideal{P,Q}}$. Furthermore, assume that they satisfy \autoref{thm:structure}\ref{case:2} and not any other case, that is, there are linear forms $v_1,v_2$ such that $T,P,Q \in \ideal{v_1,v_2}$. Finally, assume $\MS(P)\not \subseteq \MS(Q)$. Then there are linear forms $v'_1, v'_2 \in \spn{v_1,v_2}$ such that the following holds:
	\begin{itemize}
		\item $P = v'_1\ell + {v'}_2^2$ for some linear form $\ell$.
		\item $Q= v'_1u-{v'}_2^2$ for some linear form $u$.
		\item $T = {v'}_2(\ell+u) + \alpha P + \beta Q$ for some constants $\alpha,\beta \in \C$,
	\end{itemize}

	where the qualities holds up to a constant non zero factor.
\end{claim} 

\begin{proof}
	Denote $P = v_1 \ell_1 + v_2 \ell_2$, $Q = v_1u_1+v_2u_2$ and  $V = \spn{v_1,v_2,u_1,u_2}=\MS(Q)$. From the fact that $Q$ is irreducible it follows, without loss of generality, that $u_1\notin \spn{v_1,v_2}$. 
	
	\paragraph{\textbf{Case 1: $\spn{\ell_1,\ell_2}\cap V = \{\vec{0}\}$:}}
	
%
%
	
	First we note that since $P$ is irreducible we must have that $\ell_1$ and $\ell_2$ are linearly independent and therefore, $\MS(P)=\spn{v_1,v_2,\ell_1,\ell_2}$ and $\dim(\MS(P))=4$. Observe further that $Q|_{\ell_1=\ell_2=0}$ is irreducible (as $\spn{\ell_1,\ell_2}\cap V = \{\vec{0}\}$). 
	
	 
\begin{claim}	 
	There exist $\alpha\in\C$ and linear forms  $v',v'' \in \spn{v_1,v_2}$ such that $T=\alpha Q + v'\ell_1 + v''\ell_2$.
\end{claim}

\begin{proof}
Since $P \equiv_{\ell_1,\ell_2}0$ we get that $T|_{\ell_1=\ell_2=0} \in \ideal{Q|_{\ell_1=\ell_2=0}}$. Irreducibility of $Q|_{\ell_1=\ell_2=0}$ implies that for some $\alpha\in\C$ and linear forms $w'$ and $w''$, we have that $T=\alpha Q + w'\ell_1 + w''\ell_2$. Denote $T' = w'\ell_1 + w''\ell_2\in \ideal{\ell_1,\ell_2}$. 
	 As $T' = T-\alpha Q \in \ideal{v_1,v_2}$ it follows that $w'\ell_1 \equiv_{v_1,v_2}- w''\ell_2$. 
 The assumption $\spn{\ell_1,\ell_2}\cap V = \{\vec{0}\}$ implies that $w' \in \spn{\ell_2,v_1,v_2}$ and $w'' \in \spn{\ell_1,v_1,v_2}$. 
Denote $w'=\beta' \ell_2 + v'$ and $w''=\beta'' \ell_1+v''$ for $v',v''\in\spn{v_1,v_2}$. Thus, $T'=(\beta'+\beta'')\ell_1 \ell_2 + \ell_1\cdot v' + \ell_2\cdot v''$. Since $T' \in \ideal{v_1,v_2}$ it follows that $\beta'+\beta''=0$. Hence,  $T'= \ell_1\cdot v' + \ell_2\cdot v''$ as claimed.
 \end{proof}

	Since  $T' \in \sqrt{\ideal{P,Q}}\subseteq \ideal{v_2,u_1,\ell_1}$ it follows that $\ell_2\cdot v''\in \ideal{v_2,u_1,\ell_1}$. Note that since $\spn{\ell_1,\ell_2}\cap \MS(Q) = \{\vec{0}\}$ it must be the case that $v'' \in \spn{v_2,u_1,\ell_1}$. The same argument shows that $v'' \in \spn{v_2,u_1}$ and since $u_1\notin \spn{v_1,v_2}$ we get that $v''\in\spn{v_2}$. Denote $v''=\beta'' v_2$. We have that $T'-\beta''P = \ell_1\cdot(v'-\beta'' v_1) = \ell_1\cdot(\gamma v_1 + \delta v_2)$, for some $\gamma,\delta\in\C$. From the assumption that $T\not\in\spn{P,Q}$ we see that at least one among $\gamma,\delta$ is not zero. 
	
	Let $\lambda$ be such that both $\lambda \cdot \gamma  + \delta \neq 0$ and $\lambda u_1 + u_2 \not\in\spn{v_1,v_2}$ hold. As $u_1,v_1,v_2,\ell_1,\ell_2$ are linearly independent it follows that there is an assignment $\vaa$ such that  $v_2(\vaa)=\ell_1(\vaa) = 1$, and $v_1(\vaa) = -\ell_2(\vaa) = \lambda$. Observe that $P(\vaa) = 0$ and $T'(\vaa) = (\gamma \lambda+\delta) \neq 0$. Also note that our choice of $\lambda$ implies that we can choose $\vaa$ so that we also have $Q(\vaa) = 0$. Consequently, $P( \vaa)=Q( \vaa) = 0$ and $T(\vaa)=T'(\vaa)\neq 0$, in contradiction to the fact that $T \in \sqrt{\ideal{Q,P}}$.
	 
	
	\paragraph{\textbf{Case 2: $\spn{\ell_1,\ell_2}\cap V \neq \{\vec{0}\}$:}}
	
	From the requirement  $\MS(P)\not\subseteq\MS(Q)$ we can assume, without loss of generality, that $\ell_1\notin \MS(Q)$. It follows that  $v_1\notin \spn{\ell_1,u_1,v_2}$. 
As 	$\spn{\ell_1,\ell_2}\cap V \neq \{\vec{0}\}$ we have that for some $\alpha'\in \C$ and $v\in V$, $\ell_2=\alpha' \ell_1+v$. By applying a change of basis to $\spn{v_1,v_2}$ (if necessary), we can assume, without loss of generality, that $\alpha=0$. Concluding, we have that $P = v_1\ell_1 + v_2v$ and $Q= v_1u_1+v_2u_2$, where $v\in V$. Note that $v\not \in \spn{v_1}$ as otherwise $P$ would be reducible. Observe further that   $\spn{\ell_1,u_1}\cap \spn{v_1,v_2} = \{0\}$.
	
	
	
	Denote $T= 	v_1t_1+v_2t_2$. As $T \in \sqrt{\ideal{Q,P}} \subseteq \ideal{v_2,\ell_1,u_1}$, it follows that $t_1 \in \spn{v_2,\ell_1,u_1}$. Denote $t_1= \alpha \ell_1 + \beta u_1 +\gamma v_2$. Let $T' = T -\alpha P -\beta Q = v_2(\gamma v_1+t_2-\alpha v -\beta u_2)$. Let $t = \gamma v_1+t_2-\alpha v -\beta u_2$. As $T\not\in\spn{P,Q}$, we have that $t\neq 0$. The following claim implies that  $t \notin V$. 
	
	\begin{claim}\label{cla:ideal-V}
		Let $A \in \sqrt{\ideal{P,Q}}$ be a polynomial. If $A\in \C[V]$ then $A \in \ideal{Q}$.
	\end{claim}

\begin{proof}\sloppy
	 Observe that if for every $\vaa \in \cZ(Q) \setminus \cZ(v_1)$ it holds that $\vaa \in \cZ(A)$ then $\cZ(Q)= \overline{\cZ(Q) \setminus \cZ(v_1)}  \subseteq \cZ(A)$\footnote{We take closure with respect to Zariski topology.}, and thus by 
	Hilbert's Nullstellensatz it follows that $A \in \sqrt{\ideal{Q}} = \ideal{Q}$ as claimed. 
	
	So let $\vaa \in (\cZ(Q) \setminus \cZ(v_1))\setminus \cZ(A)$. As $\ell_1\not\in V$ and $v_1(\vaa)\neq 0$, we can choose $\vaa$ that in addition satisfies $P(\vaa) =0$. Thus, we have that $P(\vaa)=Q(\vaa) = 0$ but $A(\vaa) \neq 0$, contradicting the assumption $A \in \sqrt{\ideal{P,Q}}$. 
\end{proof}

Observe that if $t\in V$ then $T'=v_2\cdot t\in \sqrt{\ideal{P,Q}}$ and $\MS(T') \subseteq V$. The claim implies that $T'\in \ideal{Q}$ in contradiction (as $Q$ is irreducible and $T'\neq 0$).

We now note that $\ell_1\in\spn{V,t}$. Indeed, if this is not the case then whenever $v_1\neq 0$ we can make sure that $P$ vanishes (by substituting $\ell_1 = - v_2\cdot v/v_1$) without affecting the values of $Q$ and $T$. A similar argument as in the proof of the claim implies that in this case too we have $T\in\ideal{Q}$ in contradiction.
We can therefore assume that $t=\ell_1 + w$ for some  $w\in V$ (this is possible as we can rescale $v_2$ if needed).
 Summarizing,  we have that:
\begin{align*}
	P &= v_1\ell_1 + v_2v\\
	Q &= v_1u_1 + v_2u_2\\
	T' & = v_2(\ell_1 + w) = v_2\ell_1+v_2w
\end{align*}
Let $A = v_2\cdot P - v_1\cdot T' = v_2(v_1\cdot w-v \cdot v_2)$. Since $P$ is irreducible we have that $A\neq 0$. As $A$ depends only on the linear functions in $V$, and clearly $A\in\sqrt{\ideal{P,Q}}$, \autoref{cla:ideal-V} implies  $A\in \ideal{Q}$. Therefore, $v_2(v_1\cdot  w-v\cdot v_2)= A = aQ$ for some non zero linear form $a$. Rescaling if necessary, irreducibility of $Q$ implies that $v_2 = a$ and $  (v_1\cdot w-v\cdot v_2) = Q$. Consequently, $v_1( w- u_1) = v_2( u_2+ v)$, which in turns implies that there is a non-zero $\epsilon\in \C$ such that $\epsilon v_1 = ( u_2+ v)$ and $\epsilon v_2 = ( w- u_1)$. Hence, $u_1 =  w-\epsilon v_2$ and $u_2= \epsilon v_1 - v$. Therefore, $Q = v_1( w-\epsilon v_2) + v_2(\epsilon v_1 - v) =  v_1w - v_2v$. 
As $P,Q\in \ideal{v,v_1}$ we deduce that $T' = v_2(\ell_1 + w)\in \ideal{v,v_1}$. Hence, $v_2 \in \spn{v_1,v}$. As $v, v_2\not \sim v_1$ it follows that we can write $v = \delta v_2 + \gamma v_1$ and up to a proper rescaling $v = v_2 + \gamma v_1$. We conclude that 
\begin{align*}
P &= v_1\ell_1 + v_2(v_2 + \gamma v_1) = v_1(\ell_1 +\gamma v_2)+ v_2^2\\
Q &= v_1w - v_2(v_2 + \gamma v_1) = v_1(w -\gamma v_2)- v_2^2\\
T' &=  v_2(\ell_1 + w) = v_2(\ell_1 +\gamma v_2 + w -\gamma v_2),
\end{align*} 
as claimed.
\end{proof}

As a consequence of the claim we can deduce the following uniqueness property.

\begin{corollary}\label{cor:unique-T}
	Let $P,Q,Q', T$ be pairwise linearly independent irreducible quadratics such that $T\in \sqrt{\ideal{P,Q}}$. Let $T'$ be such that $T' \not \sim P,Q,Q'$ and such that   $T'\in \sqrt{\ideal{P,Q'}}$. Assume further that $P \in \ideal {\MS(Q)+\MS(Q')}$ but $\MS(P) \not \subseteq \MS(Q)+\MS(Q')$. Then $T\neq T'$. In addition, $\MS(T),\MS(T')\not\subseteq{\MS(Q)+\MS(Q')}$.
\end{corollary}

\begin{proof}
	Set $V ={\MS(Q)+\MS(Q')}$. If $P,Q$ satisfy \autoref{thm:structure}\ref{case:rk1} then there is a linear form $\ell$ such that $\ell^2 = \alpha P+\beta Q\in \ideal{V}$. This implies $\ell\in V$ and therefore $\MS(P) \subseteq \MS(Q)+\spn{\ell}\subseteq V$, in contradiction. Similarly $P,Q'$ do not satisfy \autoref{thm:structure}\ref{case:rk1}.

	Assume towards a contradiction that $T=T'$.  
	
	From pairwise  independence, if $T \in \spn{P,Q'}\cap \spn{P,Q}$  then $P\in \spn{Q,Q'}$, and again we get $\MS(P)\subseteq V$, in contradiction. Thus, without loss of generality, $P$ and $Q$ satisfy \autoref{thm:structure}\ref{case:2} and not any other case of \autoref{thm:structure}.  \autoref{cla:case3-strong} implies  (up to rescaling of $P,Q$ and $T$) that for some $v_1,v_2,u\in V$ and $\alpha,\beta\in\C$:  	
	\begin{align}
	P &= v_1\ell+ v_2^2 \nonumber\\
	Q & = v_1u- v_2^2 \label{eq:Pv1}\\
	T &= v_2(\ell + u) + \alpha P + \beta Q \nonumber \;,
	\end{align}
	where $\ell\not\in V$.
	If $T \in \spn{P,Q'}$ we deduce that $v_2(\ell + u) + \alpha P + \beta Q = T=\gamma P+\delta Q'$. Rearranging we obtain that $\ell(v_2+(\alpha -\gamma)v_1) \in \C[V]_2$ which implies that  either $v_1\sim v_2$, in contradiction to the assumption that the polynomials are irreducible, or, that $\ell \in V$, in contradiction to the assumption that $\MS(P)\not \subseteq  V$. Thus, it follows that  $P$ and $Q'$ satisfy \autoref{thm:structure}\ref{case:2} and not any other case of \autoref{thm:structure}. Applying  \autoref{cla:case3-strong} again we get  that for some $u_1,u_2,u'\in V$ and $\nu,\alpha',\beta'\in\C$:	
	\begin{align}
	P &= u_1\ell'+ u_2^2 \nonumber\\
	Q' & = u_1u'- u_2^2 \label{eq:nu}\\
	T &= \nu(u_2(\ell' + u')) + \alpha' P + \beta' Q'.\nonumber
	\end{align}		
	Considering both representations of $P$ we deduce that $v_1 \ell + v_2^2 \in \ideal{u_1,u_2}$. As $\ell \notin \spn{u_1,u_2,v_1,v_2}\subseteq V$ it follows that $\spn{v_1,v_2} = \spn{u_1,u_2}$. Thus for some $\delta, \delta',\gamma,\gamma'\in\C$ we have $v_1 = \delta' u_1+\gamma' u_2$ and $v_2 = \delta u_1 + \gamma u_2$. Substituting we obtain 
	$ u_1\ell'+ u_2^2 = P = v_1\ell+ v_2^2=(\delta' u_1+\gamma' u_2)\ell + (\delta u_1 + \gamma u_2)^2 $.
	Thus,
	\[ u_1\ell'+  u_2^2=u_1(\delta'\ell+\delta^2 u_1 + 2\delta\gamma u_2) +u_2(\gamma'\ell+\gamma^2u_2)\]
	and therefore 
	\begin{equation}\label{eq:beta=0}
		u_1(\ell' - (\delta'\ell+\delta^2 u_1 + 2\delta\gamma u_2)) = u_2(\gamma'\ell+(\gamma^2-1)u_2)\;.	
	\end{equation}
	Since $u_1\not \sim u_2$, it follows that $\gamma'\ell+(\gamma^2-1)u_2\sim u_1$, which implies $\gamma' = 0$ (as otherwise we will have $\ell \in V$) and $\gamma=\pm 1$. Rescaling $u_2$ and both $u_1$ and $u$ if needed, we can assume without loss of generality that $\gamma=1$ and $\delta'=1$. Concluding, we have that  
	\[v_1=u_1 \quad \text{ and } \quad v_2 = \delta u_1 + u_2 \;.
	\]
	Equation~\eqref{eq:beta=0} now gives   $u_1(\ell' - (\ell+ \delta^2 u_1 + 2 \delta u_2))=0$ which implies 	
	\begin{equation}
		 \ell' = \ell+\delta^2 u_1 + 2\delta  u_2 = \ell +\delta(v_2+u_2) \;.\label{eq:l'}
		 \end{equation}	
Similarly, considering both representations of $T$ we get
\begin{equation} \nu(u_2(\ell' + u')) + \alpha' P + \beta' Q' =  T = (\delta u_1 + u_2)(\ell + u) + \alpha  P + \beta  Q \;. \label{eq:Q=Q'}
\end{equation}
	Rearranging and using Equation~\eqref{eq:l'} gives
	\[\ell \cdot \left((\nu  - 1) u_2 + (\alpha' - \delta-\alpha )u_1\right) \in\C[V]_2\;. \]
	As $\ell\not\in V$ and $u_1\not\sim u_2$, it must hold that
	\[\nu=1 \quad  \text{ and } \quad \alpha'=\delta+\alpha \;.\] 
	Substituting in Equation~\eqref{eq:Q=Q'} we get
		 \[
u_2(\delta(v_2+u_2)+u') + \delta v_2^2+ \beta'(u_1 u' - u_2^2)  = u v_2 + \beta(u_1 u -v_2^2)\;.\]
 Rearranging, 
 	\begin{align*}	 
	u'(u_2 + \beta' u_1) - u (v_2 +\beta u_1) = -(\beta+\delta)v_2^2 - u_2(\delta (v_2+u_2)  -\beta' u_2) 
\end{align*} 
		and as the second term is a homogeneous quadratic in two linear forms it follows that it is of the form $ww'$ for $w,w'\in \C[\spn{v_2,u_2}]_2$. Therefore, there is a non trivial linear combination $u'= \epsilon  u + u_3$ for $u_3\in \spn{v_2,u_2}$. Substituting we get that 
		\[ u (\epsilon u_2+\epsilon\beta' u_1 - (\delta+\beta)u_1-u_2)\in \spn{u_1,u_2} \;.\]
		If $u\in\spn{u_1,u_2}$ then $Q$ is reducible in contradiction. Therefore, as $u_1 \not\sim u_2$, it must hold that 
		\[\epsilon=1 \; (\text{hence } u'=u+u_3), \quad \text{ and } \quad \beta'=\beta+\delta \;.\] 
		Looking again at Equation~\eqref{eq:Q=Q'} we obtain
		\begin{align*}
	u_2(\ell'+u' )+\delta(P+Q') &= v_2(\ell + u)+\beta(Q-Q')\\
	(u_2+\delta u_1)(\ell'+u' )  &= v_2 (\ell + u)+\beta(Q-Q')\\
	v_2 (\ell'+u' -\ell-u)  &=  \beta(u_1(u-u')+u_2^2-v_2^2 )\\
	v_2 (\delta (v_2  + u_2) +u_3)  &=  -\beta u_1 ( u_3+\delta (v_2+u_2) )\;.
\end{align*}
	Setting $u_1=0$ we conclude that for some $\mu$, $u_3+2\delta u_2 = \mu u_1$. Therefore, 
			\begin{align*}
	v_2 (\delta (v_2  + u_2) +u_3)  &=  -\beta u_1 ( u_3+\delta (v_2+u_2) )\\
	u_3(v_2+\beta u_1) &=-\delta(v_2+u_2)( v_2 +\beta u_1) \;.
\end{align*}
	As $v_2\not \sim u_1$ it follows that 
	\begin{align*}
		\mu u_1 - 2\delta u_2 &=u_3=-\delta(v_2+u_2)\\
		\mu u_1 &= \delta(u_2-v_2)=-\delta^2 u_1
	\end{align*}
	Hence, $\delta^2+\mu=0$. Therefore, 
	\begin{align*}
		u_1 u' &= u_1 u + u_1u_3\\
		&=u_1 u - \delta^2 u_1^2-2\delta u_1 u_2\\
		&= u_1 u -v_2^2+u_2^2\;.
	\end{align*}
	Rearranging we get 
		$Q'=u_1 u' - u_2^2 = u_1 u - v_2^2=Q$,
		in contradiction.
%
\end{proof}

We finish this section by formulating the improvement for \autoref{thm:structure} which follows immediately from \autoref{cla:case3-strong}

\begin{corollary}[Improvement of Theorem 1.10 of \cite{ShpilkaSG}]\label{cor:case-3-strong}
	Let $A,B$ and $C$ be $n$-variate, homogeneous, quadratic polynomials, over $\C$, such that $C\in\sqrt{\ideal{A,B}}$. Then, one of the following cases must hold:
	\begin{enumerate}[label={(\roman*)}]
		\item 
		$C$ is in the linear span of $A$ and $B$.  
		\item 
		There exists a non trivial linear combination of the form $\alpha A+\beta B = \ell^2$ for some linear form $\ell$. 
		\item 
		If none of the above hold, then there exist two linear forms $\ell_1$ and $\ell_2$ such that  $A, B,C\in\ideal{\ell_1,\ell_2}$. Furthermore, we have that either $\MS(P)\subseteq\MS(Q)$ or
			\begin{itemize}
				\item $A = \ell_1 a + \ell_2^2$ for some linear form $a$.
				\item $B= \ell_1b-\ell_2^2$ for some linear form $b$.
				\item $C = {\ell}_2(a+b) + \alpha A + \beta B$ for some constants $\alpha,\beta \in \C$.
			\end{itemize}
	 	\end{enumerate}
\end{corollary}

\section{Robust Sylvester-Gallai theorem for quadratic polynomials.}

We divide $\cQ = \cQ_1\cup\cQ_2\cup \cQ_3$ as following:


\begin{equation}\label{eq:Q_1}
\cQ_1 =\left\lbrace Q \in \cQ \;\middle|\;
\begin{tabular}{@{}l@{}}
Q satisfies  \autoref{thm:structure}\ref{case:span} with  at least\\$\delta/100$ fraction of the polynomials in $\cQ$
\end{tabular}
\right\rbrace,
\end{equation}
\begin{equation}\label{eq:Q_2}
\cQ_2 =\left\lbrace Q \in \cQ \;\middle|\;
\begin{tabular}{@{}l@{}}
Q satisfies  \autoref{thm:structure}\ref{case:rk1} with  at least\\$\delta/100$ fraction of the polynomials in $\cQ$
\end{tabular}
\right\rbrace,
\end{equation}
\begin{equation}\label{eq:Q_3}
\cQ_3 =\left\lbrace Q \in \cQ \;\middle|\;
\begin{tabular}{@{}l@{}}
Q satisfies  \autoref{thm:structure}\ref{case:2} with  at least\\$\delta/100$ fraction of the polynomials in $\cQ$
\end{tabular}
\right\rbrace.
\end{equation}
We will also use the following notation: Let $Q\in \cQ$, and $t\in \{(i),(ii),(iii)\}$ we denote \[\Gamma_t(Q) = \{P\in \cQ \mid Q,P \text{ satisfiy case t of \autoref{thm:structure}}\}\,.\]
Finally we set $\cQ_1 = \cQ_1 \setminus (\cQ_2\cup \cQ_3)$. This implies that if $P\in \cQ_1$ then at least a $\delta/100$ fraction of the polynomials in $\cQ$  satisfy  \autoref{thm:structure}\ref{case:span} with $P$ and no other case. 

\begin{observation}\label{obs:graph}
The definition of $\Gamma_t$ naturally defines an undirected graph with an edge between $P$ and $Q$ if for some $t$, $Q\in\Gamma_t(P)$ (which is equivalent to saying $P\in\Gamma_t(Q)$). Thus, when we speak of ``edges'' and ``neighbors'' this graph is the one that we refer to.
\end{observation}

Throughout the proof, we will use the following simple claim.

\begin{claim}\label{cla:edge-removal}
	Let $P, T\in \cQ$.  Removing $T$ from $\cQ$, causes the removal of at most two polynomials from $\Gamma_{(i)}(P)$, and this happens only in the case that  $P\in \Gamma_{(i)}(T)$ and $|\cQ\cap\spn{P,T}|=3$. 
\end{claim}

\begin{proof}
	First, note that for $Q_1,Q_2,Q_3\in \cQ$ if $Q_3\in \spn{Q_1,Q_2}$, then for every $k\neq j\in [3]$, $Q_k\in \Gamma_{(i)}(Q_j)$. In particular, if $P\not\in \Gamma_{(i)}(T)$, then removing $T$ from $\cQ$ does not affect $\Gamma_{(i)}(P)$.
	
	Let $P\in \Gamma_{(i)}(T)$. By the argument above, if $|\cQ\cap\spn{P,T}|>3$ then removing $T$ does not affect $\Gamma_{(i)}(P)$. Thus, the only case the $\Gamma_{(i)}(P)$. is affected is when $|\cQ\cap\spn{P,T}|=3$ and in this case the third polynomial in the span is removed from $\Gamma_{(i)}(P)$. 
\end{proof}


 
The proof of \autoref{thm:psg-robust} is organized as follows. In Section~\ref{sec:bound-Q2} we bound the dimension of $\cQ_2$. Specifically, we prove the following claim. 
 
\begin{claim}\label{cla:Q_2-dim}
There exist a subset $\cI\subseteq \cQ_2$ of size $|\cI| = O({1}/{\delta})$, and a linear space of linear forms $V'$ such that $\dim(V') = O({1}/{\delta^2})$ such that $\cQ_2\subset \spn{\cI, \C[V']_2}$.
\end{claim}

In Section~\ref{sec:Q3-ideal} We prove that for some small dimensional space
$V''$, it holds that $\cQ_3 \subset \ideal{V''}$.

\begin{claim}\label{cla:VforP3}
	There exists a linear space of linear forms, $V''$, such that $\dim(V'')=O({1}/{\delta})$ and $\cQ_3 \subset \ideal{V''}$.
\end{claim}

Set $V=V'+V''$. So far it holds that $\cQ_2\in \spn{\cI, \C[V]_2}$ and  $\cQ_3 \subset \ideal{V}$. Next, we find a small set of polynomials $\cJ$ such that $\cQ \subset \ideal{V}+\spn{\cJ}$. We prove this claim in Section~\ref{sec:QinJ,V}.


\begin{claim}\label{cla:Iexists}
	There exists a set $\cJ\subseteq \cQ$, of size   $|\cJ|= O(1/\delta)$, such that $\cQ \subset \spn{(\cQ\cap\ideal{V}),\cJ,\C[V]_2}$. Furthermore, if $P\in\cQ\setminus \ideal{V}$ then there is no quadratic $L$  such that  $P + L \in \ideal{V}$ and $\rank_s(L) \leq 2$.
\end{claim}

Given the claims above we have that $\cQ \subset \spn{(\cQ\cap\ideal{V}),\cJ,\C[V]_2}$, where $|\cJ|=O({1}/{\delta})$ and $\dim(V) = O({1}/{\delta^2})$.
We are not done yet as the dimension of $\ideal{V}$, as a vector space, is not a constant. 
To bound this dimension we partition $\cQ$ to four sets and study the subgraphs induces by any two of the sets.

%
%
%
%

\begin{align}
\LRV & =\left\lbrace Q \in \cQ \;\middle|\;	Q\in \C[V]_2 \right\rbrace \label{eq:Q_v} \\
\LRIdealV &=\left\lbrace Q \in \cQ \;\middle|\; Q\in \ideal{V} \right\rbrace \setminus \LRV \label{eq:Q_ideal}\\
\label{eq:v} \HRV &=\left\lbrace Q \in \cQ \;\middle|\;Q\in\spn{\cJ,\C[V]_2}\setminus \C[V]_2 \right\rbrace\\ 
\label{eq:J_ideal}
\HRIdealV &=\left\lbrace Q \in \cQ \;\middle|\;
Q\in \spn{\cJ,\ideal{V}}\setminus \ideal{V}
\right\rbrace \setminus \HRV \;.
\end{align}

In words, $\LRV$ is the set of all quadratics in $\cQ$ that only depend on linear functions in $V$. 
$\LRIdealV$ is the set of polynomials that are in $\ideal{V}$ but not in $\LRV$, etc. 

Our goal is to bound the dimension of each of these sets. In fact, we already know that $\dim(\LRV),\dim(\HRV)\leq O(1/\delta^4)$ so we only need to bound $\dim(\LRIdealV)$ and $\dim(\HRIdealV)$.
For that we will analyze the edges between the different sets.

%



We first note that the ``furthermore'' part of \autoref{cla:Iexists}, stating that the ``rank-distance'' between nonzero polynomials in $\spn{\cJ}$ and quadratics in $\ideal{V}$ is larger than $2$, implies the following:

\begin{observation}\label{obs:cut-case-1}
\begin{enumerate}
	\item 	If $P\in \LRIdealV\cup \LRV$ and $Q\in \HRIdealV\cup \HRV$ satisfy that $P\in \Gamma(Q)$ then $P$ and $Q$ satisfy \autoref{thm:structure}\ref{case:span}.
	\item 	If $P\in \HRIdealV$ and $Q\in \LRIdealV\cup \LRV\cup \HRV$ satisfy that $P\in \Gamma(Q)$ then $P$ and $Q$ satisfy \autoref{thm:structure}\ref{case:span}.
\end{enumerate}
\end{observation}

\begin{proof}
	We only prove the first case as the proof of the second case is similar.
	As $Q\in \HRIdealV\cup \HRV$, we have that $\rank_s(Q_1)>2$. In particular, $P$ and $Q$ do not satisfy \autoref{thm:structure}\ref{case:2}. If $P$ and $Q$  satisfy \autoref{thm:structure}\ref{case:rk1} then $Q = \alpha P+\ell^2$ for some linear form $\ell$, which contradicts the structure of $\cJ$ guaranteed in \autoref{cla:Iexists}.
\end{proof}

To bound the dimension of $\LRIdealV$ we note that any edge going from $P\in\LRIdealV\cup \HRIdealV$ to $\LRV \cup \HRV$ defines uniquely a third polynomial in $\LRIdealV \cup \HRIdealV$. This uniqueness property guarantees that if we add $\MS(P)$ to $V$, then many polynomials move from $\LRIdealV\cup \HRIdealV$ to $\LRV\cup\HRV$.  

\begin{claim}\label{cla:C-unique}
	Let $P\in \LRIdealV$ then, 
	\begin{enumerate}
			\item for every polynomial  $Q_1 \in\Gamma(P)\cap \HRV$ there is a unique polynomial $Q'_1\in \HRIdealV$ such that $Q'_1 \in \spn{P,Q_1}$. I.e., there is no other $Q_2\in\HRV$ such that $Q'_1 \in \spn{P,Q_2}$.
		\item for every polynomial  $Q_1 \in\Gamma(P)\cap \LRV$ there is a unique polynomial $Q'_1\in \LRIdealV$ such that $Q'_1 \in \sqrt{\ideal{P,Q_1}}$. I.e., there is no other $Q_2\in\LRV$ such that $Q'_1 \in \sqrt{\ideal{P,Q_2}}$.
	\end{enumerate}
\end{claim}
\begin{proof}
	\begin{enumerate}
		\item Let $Q_1 \in\Gamma(P)\cap \HRV$. 
		 By Observation~\ref{obs:cut-case-1}, $P$ and $Q_1$ satisfy \autoref{thm:structure}\ref{case:span}. We first prove that they span a polynomial in $\HRIdealV$ and then prove its uniqueness.
		 Any polynomial in $T\in\spn{P,Q_1}\setminus(\spn{P})$ has $\rank_s(T)>2$, even when setting the linear forms in $V$ to $0$. Hence, $P$ and $Q_1$ span a polynomial  $Q'_1\in \HRV\cup \HRIdealV$. As $P\not \in \C[V]_2$ we can conclude that $Q'_1\in\HRIdealV$. To prove that $Q'_1$ is unique assume that $Q'_1\in \spn{P,Q_2}$ for some $Q_2\in \HRV$. Pairwise linear independence implies that $P\in\spn{Q_1,Q_2}$ which implies that $P\in \LRV$, in contradiction.
		
		\item Follows from \autoref{cor:unique-T}.
	\end{enumerate}
\end{proof}

\begin{claim}\label{cla:J-unique}
	Let $P\in \HRIdealV$.
	Then for every polynomial  $Q_1 \in\Gamma(P)\cap (\HRV\cup \LRV)$ there is a unique polynomial $Q'_1\in \HRIdealV\cup \LRIdealV$ such that $Q'_1 \in \spn{P,Q_1}$. By ``unique'' we mean that there is no other $Q_2\in\HRV$ such that $Q'_1 \in \spn{P,Q_2}$.
\end{claim}

\begin{proof}
	We first consider the case $Q_1 \in\Gamma(P)\cap \LRV$. Observation~\ref{obs:cut-case-1} implies that $P$ and $Q_1$ satisfy \autoref{thm:structure}\ref{case:span}. By construction of $\cJ$, any polynomial in $T\in\spn{P,Q_1}\setminus(\spn{Q_1})$ has $\rank_s(T)>2$, even when setting the linear forms in $V$ to $0$. Hence, $P$ and $Q_1$ span a polynomial  $Q'_1\in \HRV\cup \HRIdealV$. As $P\not \in \HRV$ we conclude that $Q'_1\in\HRIdealV$. To prove that $Q'_1$ is unique assume that $Q'_1\in \spn{P,Q_2}$ for some $Q_2\in \HRV\cup \LRV$. As before, pairwise linear independence shows that $P\in\spn{Q_1,Q_2}$, which implies that $P\in \HRV$, in contradiction.

		 Consider the case $Q_1 \in\Gamma(P)\cap \HRV$. As before,  $P$ and $Q_1$ must satisfy \autoref{thm:structure}\ref{case:span}. Any polynomial in $T\in\spn{P,Q_1}\setminus(\spn{Q_1})$ is not in $\HRV\cup \LRV$.  Hence, $P$ and $Q_1$ span a polynomial  $Q'_1\in \LRIdealV\cup \HRIdealV$. Uniqueness follows exactly as in the first case.
\end{proof}

We next show that the uniqueness property proved in Claims~\ref{cla:C-unique} and~\ref{cla:J-unique} imply that $\cJ_{\ideal{V}}$ and $\cC_{\ideal{V}}$ cannot be ``too small,'' unless they are empty.

\begin{claim}\label{cla:C,J-large}
	If $|\HRIdealV|, |\LRIdealV| \leq (\delta/10)\cdot m$, then  $\cJ_{\ideal{V}}= \cC_{\ideal{V}}=\emptyset$.
\end{claim}
\begin{proof}
	Assume towards a contradiction that there is $P\in \cC_{\ideal{V}}\cup \cJ_{\ideal{V}}$. As $|\Gamma(P)|\geq \delta m$ it follows that  $|\Gamma(P)\cap(\LRV\cup\HRV)|\geq (8\delta/10)\cdot m$.  Claims~\ref{cla:C-unique} and~\ref{cla:J-unique} imply that there are at least $|\Gamma(P)\cap(\LRV\cup \HRV)|\geq  8\delta/10$ polynomials in $\HRIdealV\cup\LRIdealV$ in contradiction to the assumption that there are at most $(2\delta/10) \cdot m$ polynomials in $\cJ_{\ideal{V}}\cup\cC_{\ideal{V}}$.
\end{proof}

Thus, if we can make$|\HRIdealV|, |\LRIdealV|\leq (\delta/10)\cdot m$ without increasing $\dim(V)$ and $|\cJ|$ too much then \autoref{cla:C,J-large} would imply that $\cQ\in \spn{\cJ,\C[V]_2}$, from which the theorem would follow. We first show how to reduce $|\LRIdealV|$ and then  we reduce $|\HRIdealV|$. We will need the following easy observation.



\begin{claim}\label{cla:decreaseCideal}
	There is a linear subspace $V\subseteq V'$, of dimension $\dim(V')\leq 1/\delta^4 \cdot \dim(V)\leq 1/\delta^6$, such that $|\cC_{\ideal{V'}}|\leq \delta/10 \cdot m$.
\end{claim}
\begin{proof}
Denote $\cB_1 = \{Q\in \LRIdealV\mid |\Gamma(Q)\cap (\HRIdealV\cup\HRV)|\geq 0.1\delta m\}$ and $\cB_2 = \LRIdealV\setminus \cB_1$.
 We first bound the dimension of $\cB_2$.
	\begin{claim}\label{cla:B2}
	$\dim(\MS(\cB_2)) = O(1/\delta \cdot \dim(V))$.
\end{claim}
\begin{proof}
	For each $Q\in \cB_2$ we remove from $\Gamma(Q)$ all the polynomials from $\HRIdealV\cup \HRV$. This removes at most $0.1\delta m$ polynomials from $\Gamma(Q)$. From \autoref{cla:C-unique} we know that $Q$ satisfies \autoref{thm:structure}\ref{case:span} with each of them, and the resulting polynomial must have high-rank and thus it is not in $\LRIdealV\cup \LRIdealV$. It follows that each $Q\in\cB_2$ has at least $0.9\delta m$ neighbors $Q'$ such that the third polynomial (the one in $\sqrt{\ideal{Q,Q'}}$) belongs to $\LRV\cup\LRIdealV$.

		Let $Q_i\in \cB_2$. Observe that a polynomial from $\cB_2$ and a polynomial from $\cB_1$ can only satisfy \autoref{thm:structure}\ref{case:span} or \autoref{thm:structure}\ref{case:2}. Indeed, if they satisfy \autoref{thm:structure}\ref{case:rk1} then the spanned $\ell^2$ must belong to $V$ and we get that $Q_i\in\LRV$, in contradiction.

	Consider now the map $T_{V,\vaa}$, for a random $\vaa$,  as in \autoref{def:proj}. It maps all polynomials $Q_i\in \LRIdealV\cup \LRV$ to polynomials of the form $T(Q_i) = z\ell_i$ for some linear form $\ell_i$. Pairwise independence and \autoref{cla:still-indep} guaranty that if  $Q_i\neq Q_j\in \LRIdealV$ then $\ell_i\not \sim \ell_j$. Denote  $T_{V,\vaa}(\cB_2) = \{\ell_i\mid \exists Q_i\in \cB_2, T_{V,\vaa}(Q_i) = z\ell_i\}$. To prove the claim we show that $T_{V,\vaa}(\cB_2)$ satisfies the conditions of \autoref{thm:lin-rob-space}. This implies that $\dim(T_{V,\vaa}(\cB_2))\leq  O(1/\delta)$. The claim then follows from \autoref{cla:z-map-dimension}.

	If $|\Gamma_{\ref{case:span}}(Q_i)\cap \LRV| > \delta m/20$ then \autoref{cla:C-unique} implies that there are more than $\delta m/20$ polynomials in $\LRIdealV\cap \spn{Q_i,\C[V]_2}$. Similarly, if $|\Gamma_{\ref{case:2}}(Q_i)\cap \LRV| > \delta m/20$ then \autoref{cla:case3-strong} implies  that there are at least $\delta m/20$ unique polynomials in $\LRIdealV$ that are in  $\C[V+\MS(Q_i)]_2$. In either cases, there are at least $\delta m/20$ different polynomials in $\LRIdealV\cap \C[V+\MS(Q_i)]_2$.
	
	The previous argument implies that for $Q_i\in \cB_2$, if $|\Gamma(Q_i)\cap \LRV| > \delta m/10 $ then there are at least $\delta m/20$ different polynomials in  $\LRIdealV\cap \C[V+\MS(Q_i)]_2$. It follows  that $\spn{\ell_i, z}$ contains at least $\delta m/20$  of the other $\ell_j$, each of them is no a multiple of $z$ (as they come from polynomials in $\LRIdealV$). Thus, there are at most $20/\delta$ such $\ell_i$s that are linearly independent modulo $z$ (i,e, that their projections on $\spn{z}^\perp$ are linearly independent). Denote their span by $W$ (thus $\dim(W)\leq 20/\delta$) and update $T_{V,\vaa}(\cB_2) = T_{V,\vaa}(\cB_2) \setminus W$. Next, remove all edges between polynomials in $T_{V,\vaa}(\cB_2)$ and $T_{V,\vaa}(\cB_1)$. This removes at most $2\delta m/10$ neighbors for each $Q_i\in \cB_2$ (uniqueness implies that each $Q\in\LRV$ can affect the removal of at most two edges as in \autoref{cla:edge-removal}).
	
	At this point we clearly have that for every $Q_i\in \cB_2$, $|\Gamma(T_{V,\vaa}(Q_i))\cap  T_{V,\vaa}(\cB_2)| > \delta m/2$. Therefore, the set $\{z, T_{V,\vaa}(\cB_2)\}$ satisfies the condition of \autoref{thm:lin-rob-space} with $W = \spn{W,z}$ and $\cW=W \cap T_{V,\vaa}(\cB_2)$. It follows that  $\dim(T_{V,\vaa}(\cB_2))\leq O(\dim(W)+1/\delta)=O(1/\delta)$.  \autoref{cla:z-map-dimension} implies that 	$\dim(\MS(\cB_2)) = O(1/\delta\cdot \dim(V))$.
\end{proof}

We next bound $|\cB_1|$. Assume that $|\cB_1|\geq \delta/10 \cdot m$. It follows that at least $\frac{ \delta^2}{100} m^2$  edges 
have one side in $\cB_1$ and the other in $\HRIdealV\cup \HRV$. Set $\cB'_1=\emptyset$. The combination of \autoref{obs:cut-case-1} and \autoref{cla:case-1-cut} with $\cB = \cB_1$ (and $\epsilon=\delta^2/100$) imply  that there is a subset $\cB' \subseteq \cB_1$ of size $O(m/\delta^2)$ and dimension $O({1}/{\delta^2})$. Set $\cB'_1 = \cB'_1 \cup \cB'$, $\cB_1 = \cB_1\setminus \cB'_1$ and repeat this again (i.e. move a set of size  $O(m/\delta^2)$ and dimension $O({1}/{\delta^2})$ from $\cB_1$ to $\cB'_1$).
As in each step we remove at least $\Omega({\delta^2}m)$ polynomials from $\cB_1$, this process must terminate after $O({1}/{\delta^2})$ many iterations.
Thus, when the process terminates, $|\cB_1| \leq \delta/10 \cdot m$ and  $\dim(\cB'_1)\leq O(1/\delta^2\cdot 1/\delta^2)= O(1/\delta^4)$. Hence, $\dim(\MS(\cB'_1))\leq O(1/\delta^4\cdot \dim(V))$.

We are now ready to define $V'$. Let $V'=\MS(\cB_2)+\MS(\cB'_1)+V$.
By \autoref{cla:B2} $\dim(\MS(\cB_2)) = O(1/\delta \cdot \dim(V))$, and as we just proved, $\dim(\MS(\cB'_1))\leq O(1/\delta^4\cdot \dim(V))$. Thus, $\dim(V')\leq O(1/\delta^4\cdot \dim(V))$. It is also clear that now $\cB_2\in\C[V]_2$ and that $|\cB_1|\leq \delta m /10$, as claimed.  	
\end{proof}

 Note that it may now be the case that some linear combination of polynomials in $\cJ$ is now ``close'' to $V'$. We therefore perform the following simple process (as in the proof of the ``furthermore'' part in \autoref{sec:QinJ,V}): if $Q\in\spn{\cJ}$ is such that for some quadratic $L$ of $\rank(L)=2$ we have that $P + L \in \ideal{V'}$ then we can add $\MS(L)$ to $V'$ and remove one polynomial from $\cJ$ while still maintaining that $\cQ\subset \spn{(\cQ\cap\ideal{V'}),\cJ,\C[V']_2}$. As $|\cJ|= O(1/\delta)$, this does not have much affect on the dimension of $V'$, which is still $O(1/\delta^4\cdot \dim(V))$.
	 
To simplify notation, we denote with $V$ the linear space guaranteed by \autoref{cla:decreaseCideal}.  As $V$ may have changed, we update the sets $\LRV$, $\LRIdealV$, $\HRV$ and $\HRIdealV$ accordingly. By construction of $V=V'$, we now have that $|\LRIdealV|\leq \delta/100 m$.

 We now complete the proof of \autoref{thm:psg-robust} by bounding the dimension of $\HRIdealV$.

\begin{claim}\label{cla:decreaseJideal}
	There is a set $\cJ\subseteq \cJ'\subset \cQ$ such that $|\cJ'|\leq |\cJ|+O(1/\delta)$ and $\dim(\cJ'_{\ideal{V}}) \leq O({1}/{\delta} + \dim(V)^2)$.
\end{claim}

\begin{proof}

 Denote $\cT_1 = \{Q\in \cJ\mid |\Gamma_{(ii)}(Q)|\geq 0.1\delta m\}$ and $\cT_2 = \HRIdealV\setminus \cT_1$. For every polynomial in $Q\in\HRIdealV$, denote $Q=Q_{\cJ}+Q_{\ideal{V}}$ where $Q_{\cJ}\in \spn{\cJ}$ and $Q_{\ideal{V}}\in \ideal{V}$. Note that neither $Q_{\cJ}$ nor $Q_{\ideal{V}}$ can  be zero as this would imply  $Q\in \HRV \cup \LRIdealV$.
	\begin{claim}
		 There is a subset $\cT'_1\subseteq \cT_1$ of size at most $10/\delta$ such that $\cT_1\subset \spn{\cT'_1,\cJ,\C[V]_2}$. 
	\end{claim}
	\begin{proof}
		Let $Q_1\neq Q_2\in \cT_1$. If $P \in \Gamma_{\ref{case:rk1}} (Q_1)\cap \Gamma_{\ref{case:rk1}}(Q_2)$ then $P= Q_1 + \ell_1^2 = \alpha_2Q_2+\ell_2^2$, for some $0\neq \alpha_2\in \C$ and linear functions $\ell_i$. As no polynomial in $\spn{\cJ}$ is at ``rank-distance'' two from $\ideal{V}$, it follows that $P\in \left(\cQ \cap \left(\spn{\cJ,\ideal{V}}\setminus{\ideal{V}}\right)\right)$. Denote $P = P_{\cJ}+P_{\ideal{V}}$ where $0\neq P_{\cJ}\in \spn{\cJ}$ and $P_{\ideal{V}}\in \ideal{V}$. Thus $P_{\cJ}-{Q_1}_{\cJ} = {Q_1}_{\ideal{V}}-P_{\ideal{V}} + \ell_1^2$. As there is no linear combination of polynomials in $\cJ$ that has $\rank_s\leq 2$ when setting $V=0$, we conclude that $P_{\cJ}-{Q_1}_{\cJ}=0$ and thus $\ell_1\in V$. Similarly, can show that $\ell_2\in V$. Thus, $Q_2\in \spn{Q_1, \C[V]_2}$. 
		
		Let $\cT'_1\subseteq \cT_1$ be a maximal subset (with respect to inclusion) that has the property that for every pair of polynomials $Q_1,Q_2\in \cT'_!$ it holds that $ \Gamma_{{\ref{case:rk1}}}(Q_1)\cap \Gamma_{\ref{case:rk1}}(Q_2) = \emptyset$. By the argument above, $\cT_1\subset \spn{\cT'_1,\C[V]_2}$. The assumption that  every $Q_i\in \cT_1$ satisfies $| \Gamma_{\ref{case:rk1}}(Q)|\geq \delta/10 m$ implies that $|\cT'_1|\leq 10/\delta$, as claimed.
	\end{proof}
	
	Set $\cT_2 = \cT_2\setminus \spn{\cT_1,\cJ,\C[V]_2}$.
	Every $Q\in \cT_2$ must now satisfy that $|\Gamma_{\ref{case:span}}(Q)|\geq 0.9\delta m$. Indeed, this follows from the fact that $Q\not\in\cT_1$ and that it cannot satisfy \autoref{thm:structure}\ref{case:2} with any polynomial. Remove from $\Gamma_{\ref{case:span}}(Q)$ all the polynomials in $\cB_1$, this removes at most $2|\cB_1|\leq 2/10 \delta m$ polynomials from $\Gamma_{\ref{case:span}}(Q)$ (using an argument similar to \autoref{cla:edge-removal}), leaving $|\Gamma_{(i)}(Q)|\geq 0.7\delta m$.  This implies that $\cK = \cT_2$, $W=\spn{\cT_1,\cJ,\C[V]_2}$ and $\cW= \cQ\cap \spn{\cT_1,\cJ,\C[V]_2}$ satisfy the conditions of \autoref{thm:lin-rob-space}. As $\dim(W)\leq O( \dim(V)^2)$ it follows  that $\dim(\HRIdealV)\leq O({1}/{\delta} + \dim(V)^2)$.
	
	Setting $\cJ'=\cT'_1\cup\cJ$ completes the proof.
\end{proof}
	
We now put everything together and prove \autoref{thm:psg-robust}.

\begin{proof}[Proof of  \autoref{thm:psg-robust}]
	Claims~\ref{cla:Q_2-dim},~\ref{cla:VforP3} and~\ref{cla:Iexists} imply that there  exists a set $\cJ\subseteq \cQ$, of size   $|\cJ|= O(1/\delta)$, and a subspace of linear functions $V$ of dimension $\dim(V)=O(1/\delta^2)$ such that  $\cQ \subset \spn{(\cQ\cap\ideal{V}),\cJ,\C[V]_2}$. 
	
	By Claims~\ref{cla:decreaseCideal} and~\ref{cla:decreaseJideal} there are $\cJ\subseteq \cJ'$ and $V\subseteq V'$ such that $\dim(V')\leq 1/\delta^4 \cdot \dim(V)\leq 1/\delta^6$ and $|\cJ'|=O(1/\delta)$, for which it holds that  $|\cC_{\ideal{V'}}|\leq \delta/10 \cdot m$ and $\dim(\cJ'_{\ideal{V}}) \leq O({1}/{\delta} + \dim(V)^2)=O(1/\delta^8)$. We now set $\cJ=\cJ'$, $V=V'$ and, if needed, we add $O(|\cJ|)$ linear functions to $V$ to make sure that no non-trivial linear combination of polynomials in $\cJ$ is of the form $L+F(V)$ where $\rank_s(L)\leq 2$ and $F\in\C[V]_2$, we obtain that $\HRIdealV=\emptyset$ and $|\LRIdealV|\leq \delta m/10$. 	
	Claim~\ref{cla:C,J-large} now guarantees that we also have that $\LRIdealV=\emptyset$. Hence, $\cQ= \LRV\cup\HRV$ and it follows that $\dim(\spn{\cQ})\leq |\cJ|+\dim(V)^2 = O(1/\delta^{16})$. 	
\end{proof}

\subsection{Poof of Claim~\ref{cla:Q_2-dim}}\label{sec:bound-Q2}

Here we study a special case in which the set of polynomials is the union of two  sets such that each polynomial from the first set satisfies \autoref{thm:structure}\ref{case:rk1} with $\delta$ fraction of the polynomials in the second set. The Claim below clearly implies \autoref{cla:Q_2-dim}.


\begin{claim}\label{cla:case_rk1-dim}
	Let $\delta\in (0,1]$. Let $\cA$ and $\cB$ be two sets of irreducible quadratic polynomials such that the polynomials in $\cA\cup\cB$ are pairwise linearly independent. Assume further that for every $Q\in \cA$ the set 
	$\Gamma(Q) = \{ P \in \cB \mid P \text{ and } Q \text{ satisfy } \autoref{thm:structure}\ref{case:rk1}\}$
	has size $|\Gamma(Q)|\geq \delta \cdot |\cB|$.
Then, there is a subset $\cI\subseteq \cA$ of size $|\cI| = O(\frac{1}{\delta})$, and a linear space of linear forms $V$ with $\dim(V) = O\left({1}/{\delta^2}\right)$, such that $\cA\subset \spn{\cI, \C[V]_2}$.
\end{claim}

\begin{proof}
	Let $\cI\subseteq \cA$ be defined as a maximal set of polynomials in $\cA$ that satisfy the following condition: For every $Q\in\cI$ there  is at most one polynomial $P\in \Gamma(Q)$ such that $P\in \Gamma(Q')\cap \Gamma(Q'')$ for two other polynomials $Q',Q''\in \cI\setminus \{Q\}$.	
In other words, for every polynomial $Q$ in $\cI$ there is at most one polynomial $P\in \Gamma(Q)$ that satisfies \autoref{thm:structure}\ref{case:rk1} with (at least) two other polynomials in $\cI$.
	\begin{claim}\label{cla:sizeI}
		For every such $\cI$ it holds that $|\cI|< \frac{6}{\delta}$.
	\end{claim}
	
	\begin{proof}
		We start with four buckets of polynomials $\cB_0= \cB$ and $\cB_1=\cB_2=\cB_3=\emptyset$.  Let us add the polynomials to $\cI$ one by one. When adding $Q$ to  $\cI$,  we  move the polynomials in $\Gamma(Q)$ from $\cB_i$ to $\cB_{i+1}$ (unless they are already in $\cB_3$, in which case we leave them there). 
		By definition of $\cI$, we know that in each step there is at most  one polynomial $P\in \Gamma(Q)\cap \cB_3$. I.e. at most one polynomial in $\Gamma(Q)$ is not being moved. As at each step $Q$ moves at least $\delta |\cB| -1> \delta/2 |\cB|$ polynomials and every polynomial can be moved at most three times, we conclude that the number of steps, which is $|\cI|$ satisfies $|\cI| < \frac{6}{\delta}$. 
	\end{proof}

	Let $\cI\subseteq \cA$ be such maximal set (with respect to inclusion).	We partition $\cB$ as in the proof of \autoref{cla:sizeI}: $P\in \cB_t$ if  $P\in \Gamma(Q_i)$ for  $t$ different polynomials $Q_i\in \cA$, and, abusing notation,  we set $\cB_3 = \cup_{t\geq 3} \cB_t$.
	\begin{claim}\label{cla:b_2b3}
		There is a linear space of linear forms, $V$, such that $\dim(V) = O({1}/{\delta^2})$ and every polynomial in $P\in \cB_2 \cup \cB_3$  is of the form $P= P' + \ell^2$ for some $\ell \in V$ and $P'\in \spn{\cI}$.
	\end{claim} 
In particular the claim implies that $\cB_2 \cup \cB_3 \subset \spn{\cI}+\C[V]_2$.
	\begin{proof}
		Note that $\cB_2 \cup \cB_3 = \cup_{Q_i\neq Q_j \in \cI} (\Gamma(Q_i)\cap \Gamma_(Q_j))$. As $|\cI|< {6}/{\delta}$, there are  $<{18}/{\delta^2}$ pairs $Q_i\neq Q_j\in \cI$. 
		We next prove that for every such $Q_i\neq Q_j\in \cI$ there is a subspace $V_{i,j}$ such that $\dim(V_{i,j})\leq 4$ and every  $P\in \Gamma(Q_i)\cap \Gamma(Q_j)$ is of the form $Q_i +\ell^2$ for some $\ell\in V_{i,j}$. This clearly implies the claim by setting $V= \cup_{i,j} V_{i,j}$.
		
		To ease notations we show the existence of $V_{1,2}$.
		We shall use the following definition of $\beta_i$, $\ell_i$ and $\ell'_i$: (possibly after rescaling) express each $P_i \in \Gamma(Q_1)\cap \Gamma(Q_2)$ as $P_i =Q_1 + {\ell_i}^2 = \beta_i Q_2 +{\ell'_i}^2$, where $\ell_i$ and $\ell'_i$ are linear forms and $\beta_i\in\C$.  We consider two cases:
		\begin{enumerate}
			\item  There is a single $\beta$ such that for every $P_i \in \Gamma(Q_1)\cap \Gamma(Q_2)$, $\beta_i = \beta$:
			Set $V_{1,2} = \spn{\ell_1,\ell_1'}$. From the two representations of $P_1$ we get that
			\begin{equation*}
			Q_1 - \beta Q_2 = \ell_1^{'2} -\ell^2_1 = (\ell_1'- \ell_1) \cdot (\ell_1' + \ell_1) \neq 0,
			\end{equation*}
			where  the expression above is nonzero as $Q_1$ and $Q_2$ are linearly independent. Consider a different $P_j$. Considering the two representations of $P_j$ we get that
			\begin{equation*}
			Q_1 - \beta Q_2 = \ell_j^{'2} -\ell^2_j = (\ell_j'- \ell_j) \cdot (\ell_j' + \ell_j) \neq 0 \;.
			\end{equation*}
			Thus,
			$  (\ell_j'- \ell_j) \cdot (\ell_j' + \ell_j) = (\ell_1'- \ell_1) \cdot (\ell_1' + \ell_1) $ .
			Unique factorization implies that $\ell_j', \ell_j\in V_{1,2}$, which is what we wanted to prove.
			
			\item There is $j$ such that $\beta_j \neq  \beta_1$:
			In this case we have that
			\begin{align*}
				P_1 =Q_1 + \ell_1 &= \beta_1Q_2 +\ell'_1\\P_j =Q_1 + \ell_j &= \beta_jQ_2 +\ell'_j \;,
			\end{align*}
			and the matrix $
			\begin{bmatrix}
			1 & -\beta_1\\
			1 & -\beta_j
			\end{bmatrix}$
			has full rank. It follows that
			\[Q_1, Q_2\in  \spn{(\ell_j'- \ell_j) \cdot (\ell_j' + \ell_j), (\ell_1'- \ell_1) \cdot (\ell_1' + \ell_1)}\;.\]
			Set
			$V_{1,2} = \spn{\ell_j',\ell_j, \ell_1', \ell_1}$.
			Consider any other $P_k\in \Gamma(Q_1)\cap\Gamma(Q_2)$ and, without loss of generality, assume $\beta_k \neq  \beta_1$. As before, 
			\[Q_1, Q_2\in  \spn{(\ell_j'- \ell_j) \cdot (\ell_j' + \ell_j), (\ell_k'- \ell_k) \cdot (\ell_k' + \ell_k)}\;.\]
			Hence, either for some $\alpha\in\C$, $Q_1 = \alpha ( \ell_j'- \ell_j) \cdot (\ell_j' + \ell_j)$ or there exist $\gamma\in\C$ and a nonzero $\delta\in\C$ such that
			$Q_1 = \gamma ( \ell_j'- \ell_j) \cdot (\ell_j' + \ell_j) + \delta (\ell_k'- \ell_k) \cdot (\ell_k' + \ell_k)$.
			In the first case we get that both $Q_1$ and $Q_1-\beta_j Q_2$ are in $\spn{( \ell_j'- \ell_j) \cdot (\ell_j' + \ell_j)}$ which implies that $Q_1,Q_2$ are linearly dependent, in contradiction to our assumption.
			In the second case we have that for some $\mu,\eta$:
			\[\gamma ( \ell_j'- \ell_j) \cdot (\ell_j' + \ell_j) + \delta (\ell_k'- \ell_k) \cdot (\ell_k' + \ell_k) = Q_1 = \mu  (\ell_j'- \ell_j) \cdot (\ell_j' + \ell_j) + \eta (\ell_1'- \ell_1) \cdot (\ell_1' + \ell_1).\]
			It follows that both $\ell_k'- \ell_k $ and $ \ell_k' + \ell_k$ are spanned by the functions
			in $V_{1,2}$. 
		\end{enumerate}
	Thus, in wither cases $V_{1,2}$ has the required property.
	\end{proof}
	
	\begin{claim}\label{cla:IfarV}
		There is a linear space of linear forms $V'$, such that $\dim(V)=O({1}/{\delta^2})$ such that 
		$\cA \subset \spn{\cI, \C[V']_2}$ and hence $\dim(\cA) = O({1}/{\delta^4})$. 
	\end{claim} 
	\begin{proof}
		We construct $V'$ from $V$ and $\cI$ in the following way: We start by setting $V'=V$. Now, as long as there is a linear combination of polynomials in $\cI$ that gives a polynomial of the form $F(V') + \ell_1\cdot \ell_2$, for $F\in\C[V']_2$, we remove one of the polynomials in the linear combination from $I$ and add $\ell_1$ and $\ell_2$ to $V'$. That is, if
		\[\sum_{Q_i\in\cI}\alpha_i Q_i = F(\vx) + \ell_1\cdot \ell_2\;,\]
		where $F\in\C[V']_2$ then we remove the first polynomial with a nonzero coefficient from $\cI$ and define (abusing notation) $V'=V'+\spn{\ell_1,\ell_2}$. Note that by doing so, $\spn{\cI}+\C[V']_2$ can only increase. As this process can take at most $|\cI|<6/\delta$ steps and at each step we added at most two linear functions to $V'$ we have that $\dim(V')\leq \dim(V)+12/\delta = O(1/\delta^2)$.

		Now, let $Q \in \cA \setminus  \cI$. As $Q\not\in\cI$, by definition of $\cI$,  there must be at least two polynomials $P_1,P_2 \in \Gamma(Q) \cap (\cB_2\cup \cB_3)$. As  $P_1,P_2 \in \Gamma(Q)$,  for $i\in[2]$,  there is a linear function $\ell'_i$ such that $P_i = Q + \ell^{'2}_i$. Note that $\ell'_1 \neq \ell'_2$ as $P_1\neq P_2$. Furthermore,  \autoref{cla:b_2b3} implies that,  for $i\in[2]$, there are $\alpha_i\in\C$, $Q_i\in \cI$ and linear form $\ell_i\in V'$, such that $P_i = \alpha_i Q_i + \ell^2_i$. Thus,  
		\[0 \neq P_1-P_2 = \ell^{'2}_1-\ell^{'2}_2 = \alpha_1 Q_1 + \ell^2_1 - \alpha_2 Q_2 - \ell^2_2 \in \spn{\cI, \C[V']_2}\;.\] 
		It follows that there is a linear combination of polynomials in $\cI$ that equals a polynomial of the form $F(V')+(\ell'_1-\ell'_2)(\ell'_1+\ell'_2)$. Thus, by construction of $V'$ it must be the case that the linear combination is trivial. In particular, $F(V')+(\ell'_1-\ell'_2)(\ell'_1+\ell'_2)=0$, which implies that $\ell'_1\pm\ell'_2 \in V'$. 
		Therefore, $Q =  \alpha_1 Q_1 + \ell^2_1 -\ell^{'2}_1 \in \spn{\cI,\C[V']_2}$ as claimed.
	\end{proof}
		This concludes the proof of \autoref{cla:case_rk1-dim}.
\end{proof}

\begin{proof}[Proof of Claim~\ref{cla:Q_2-dim}]
	The claim  follows from \autoref{cla:case_rk1-dim} for $\cA = \cQ_2$, $\cB=\cT$ and $\delta= \delta/100$.
\end{proof}

\subsection{Proof of Claim~\ref{cla:VforP3}}\label{sec:Q3-ideal}

Here we consider the situation where the set of polynomials is the union of two  sets such that each polynomial in the first  set satisfies \autoref{thm:structure}\ref{case:2} with $\delta$ fraction of the polynomials in the second set. The proof of \autoref{cla:VforP3} is an immediate corollary of the next claim.


\begin{claim}\label{cla:case2-ideal}
	Let $\delta\in (0,1]$. Let $\cA$ and $\cC$ be two sets of irreducible quadratic polynomials such that the polynomials in $\cA\cup\cC$ are pairwise linearly independent. Furthermore, assume that for every  $Q\in \cA$ there is a subset $\Gamma(Q)\subseteq \cC$ such that $ |\Gamma(Q)|\geq \delta|\cC|$ and for every $P\in \Gamma(Q)$, $Q$ and $P$ satisfy \autoref{thm:structure}\ref{case:2}. Then, there is a linear space of linear forms $V$ such that $\dim(V) = O({1}/{\delta})$ and $\cA\subset \ideal{V}$.
\end{claim}

\begin{proof}

The intuition  behind the claim is based on the following observation.

\begin{observation}\label{rem:4-2-dim}
	If $Q\in \cA$ and $P\in \cC$ satisfy \autoref{thm:structure}\ref{case:2} then $\dim(\MS(Q)), \dim(\MS(P)) \leq 4$ and $\dim(\MS(Q)\cap\MS(P)) \geq 2$. 
\end{observation}

Thus, we have many small dimensional spaces that have large pairwise intersections and we can therefore expect that such a $V$ may exist.

	We prove the existence of $V$ by explicitly constructing it.
	Consider the following process: start by setting $V' = \{\vec{0}\}$, and $\cA'=\emptyset$. As long as there is $Q \in \cA$, such that $Q\notin \ideal{V'}$, we set $V' =\MS(Q) + V'$, and $\cA'=\cA' \cup \{Q\}$.  We  show next that this process must end after at most $\frac{3}{\delta}$ steps. In particular, at termination, $|\cA'| \leq \frac{3}{\delta}$, $\dim(V')\leq 12/\delta$ and $\cA\subset\ideal{V'}$.

	
	\begin{claim}\label{cla:3-case3}
		Let $P\in \cC$. Let $\cB\subseteq \cA'$ be the subset of all polynomials in $\cA'$ that satisfy \autoref{thm:structure}\ref{case:2} with $P$. Then, $|\cB| \leq 3$.
	\end{claim}

	\begin{proof}
		Assume towards a contradiction that $|\cB| \geq 4$, and that $Q_1,Q_2,Q_3$ and $Q_4$ are the first four elements of $\cB$ that where added to $\cA'$, in this order. Denote $U=\MS(P)$, and  $U_i = U\cap \MS(Q_i)$, for $1\leq i \leq 4$.
		
		As $P$ and $Q_1$ satisfy \autoref{thm:structure}\ref{case:2}, we have that  $\dim(U) \leq 4$. Furthermore, for every $i$, $\dim(U_i)\geq 2$ (by \autoref{rem:4-2-dim}). As the $Q_i$s were picked by the iterative process, we have that $U_2 \not \subseteq U_1$. Indeed, since $Q_2 \in \ideal{U_2}$, if we had $U_2 \subseteq U_1\subseteq  \MS(Q_1)\subseteq V'$, then this would imply that $Q_2\in\ideal{V'}$, in contradiction to the fact that $Q_2$ was added to $\cA'$ after $Q_1$. Similarly we get that
		$U_3 \not \subseteq U_1 + U_2$ and  $U_4 \not \subseteq U_1+U_3 +U_3$. However, as the next simple lemma shows, this is not possible.
		\begin{lemma} \label{lem:3-are-V}
			Let $V$ be a linear space of dimension $\leq 4$, and let $V_1,V_2,V_3 \subset V'$ each of dimension $\geq 2$, such that $V_1\not \subseteq V_2$ and  $V_3\not \subseteq V_2 + V_1$. Then, $V = V_1+V_2+V_3$.
		\end{lemma}
		\begin{proof}
			As $V_1\not \subseteq V_2$ we have that $\dim(V_1+V_2)\geq 3$. Similarly we get 
			$4\leq \dim(V_1+V_2+V_3)\leq \dim(V)=4$.
		\end{proof}
		\autoref{lem:3-are-V} implies that $V'=U_1+U_2+U_3$ and in particular, $U_4 \subseteq U_1+U_2+U_3$ in contradiction. This completes the proof of \autoref{cla:3-case3}.
	\end{proof}
	On the one hand we have that for every $Q \in \cA'$, it holds that $|\Gamma{(Q)}|\geq \delta|\cC|$. On the other hand, \autoref{cla:3-case3} implies that each $P\in \cC$  satisfies \autoref{thm:structure}\ref{case:2} with at most three different polynomials in $\cA'$. It follows that $|\cA'| \cdot  \delta|\cC| \leq 3\cdot  |\cC|$ and therefore,  $|\cA'|\leq 3/\delta$.
	As in each step we add at most four linearly independent linear forms to $V'$, we obtain $\dim(V')\leq \frac{12}{\delta}$.
\end{proof}

\begin{proof}[Proof of Claim~\ref{cla:VforP3}]
	The claim  follows from \autoref{cla:case2-ideal} for $\cA = \cQ_3$, $\cC=\cT$ and $\delta= \delta/100$.
\end{proof}

\subsection{Proof of Claim~\ref{cla:Iexists}}\label{sec:QinJ,V}

\begin{proof}
	As in the proof of \autoref{cla:case_rk1-dim}, we construct $\cJ$   iteratively as follows. 
	We start by setting $\cJ= \cI$ and $\cB=\cQ_2\cup\cQ_3$. We first add to $\cB$ any polynomial from $\cQ_1$ that is in $ \spn{(\cQ\cap\ideal{V}),\cJ,\C[V]_2}$. Observe that at this point we have that $\cB\subset  \spn{(\cQ\cap\ideal{V}), \C[V]_2,\cJ}$. Next, consider the following iterative process for the polynomials in $\cQ_1$:
	In each step pick any $P \in \cQ_1\setminus \cB$ that satisfies \autoref{thm:structure}\ref{case:span}  with at least $\frac{\delta}{300}m$ polynomials in $\cB$, and add it to both $\cJ$ and to $\cB$. Then, add to $\cB$ all the polynomials $P'\in\cQ_1$ that satisfy $P' \in \spn{(\cQ \cap \ideal{V}), \cJ, \C[V]_2}$. Note, that we always maintain that  $\cB \subset \spn{(\cQ\cap\ideal{V}), \C[V]_2,\cJ }$.	
	We continue this process as long as possible.
	The next claim shows that  when the process terminates we have that $|\cJ| \leq O(1/\delta)$.
	
	\begin{claim}\label{cla:2nd-proc}
		In each step we added to $\cB$ at least $\frac{\delta}{300}m$ new polynomials from $\cQ_1$. In particular, $|\cJ| \leq 300/\delta$.
	\end{claim}
	\begin{proof}
		Consider what happens when we add some polynomial $P$ to $\cJ$. By the description of our process, $P$ satisfies \autoref{thm:structure}\ref{case:span} with at least $\frac{\delta}{300}m$ polynomials in $\cB$. Any $Q\in  \cB$, that satisfies \autoref{thm:structure}\ref{case:span} with $P$,  must span with $P$ a polynomial $P'\in {\cQ_1}$. To see that  $P'\in \cQ_1$ observe that  otherwise we would have that $P\in \spn{\cB}\subset  \spn{(\cQ\cap\ideal{V}), \C[V]_2,\cJ }$, which implies $P\in\cB$, in contradiction to the definition of the process.
		Furthermore, for each such $Q\in \cB$ the polynomial $P'$ is unique. Indeed, if there was a $P\neq P'\in\cQ_1$ and $Q_1,Q_2\in  \cB$ such that $P'\in \spn{Q_1,P}\cap\spn{Q_2,P}$ then pairwise independence implies that $P\in \spn{Q_1,Q_2}\subset \spn{\cB}$, from which we infer that $P\in\cB$, in contradiction. Thus, when we add $P$ to $\cJ$, at least $\frac{\delta}{3}m$ polynomials are added to $\cB$. In particular, the process must terminates after at most $300/\delta$ steps. Consequently, $|\cJ|\leq 300/\delta$ as claimed.
	\end{proof} 
	
	
	Consider the polynomials left in $\cQ_1\setminus \cB$ when the process terminated. As they ''survived'' the process, each of them satisfies \autoref{thm:structure}\ref{case:span} with less than $\frac{\delta}{300}m$ polynomials in $\cB$. From the fact that $\cQ_3\cup \cQ_2\subseteq \cB$ and  the uniqueness property that was observed in the proof of \autoref{cla:2nd-proc}, we get that the set $\cQ_1\setminus \cB$ satisfies the conditions of \autoref{def:delta-SGConf} with parameter $\delta'=(\delta/100-\delta/300)$.  \autoref{thm:DSW} implies that $\dim(\cQ_1\setminus \cB)\leq O(1/\delta)$. Adding a basis of $\cQ_1\setminus \cB$ to $\cJ$ we get that $|\cJ| = O(1/\delta)$ and every polynomial in $\cQ$ is in $\spn{(\cQ\cap\ideal{V}), \C[V]_2,\cJ }$.

To prove the furthermore part we note that if $Q\in\cQ$ is such that for some quadratic $L$ of $\rank(L)=2$ we have that $P + L \in \ideal{V}$ then we can add $\MS(L)$ to $V$ and remove one polynomial from $\cJ$ while still maintaining that $\cQ\subset \spn{(\cQ\cap\ideal{V}),\cJ,\C[V]_2}$, in a similar way to the process described at \autoref{cla:IfarV}.

\end{proof}

\bibliographystyle{customurlbst/alphaurlpp}
\bibliography{bibliography}

\begin{thebibliography}{BDWY13}

\bibitem[Agr05]{DBLP:conf/fsttcs/Agrawal05}
Manindra Agrawal.
\newblock \href {http://dx.doi.org/10.1007/11590156\_6} {Proving Lower Bounds
  Via Pseudo-random Generators}.
\newblock In Ramaswamy Ramanujam and Sandeep Sen, editors, {\em {FSTTCS} 2005:
  Foundations of Software Technology and Theoretical Computer Science, 25th
  International Conference, Hyderabad, India, December 15-18, 2005,
  Proceedings}, volume 3821 of {\em Lecture Notes in Computer Science}, pages
  92--105. Springer, 2005.

\bibitem[AV08]{DBLP:conf/focs/AgrawalV08}
Manindra Agrawal and V.~Vinay.
\newblock \href {http://dx.doi.org/10.1109/FOCS.2008.32} {Arithmetic Circuits:
  {A} Chasm at Depth Four}.
\newblock In {\em 49th Annual {IEEE} Symposium on Foundations of Computer
  Science, {FOCS} 2008, October 25-28, 2008, Philadelphia, PA, {USA}}, pages
  67--75. {IEEE} Computer Society, 2008.

\bibitem[BDSW19]{DBLP:journals/dcg/BasitDSW19}
Abdul Basit, Zeev Dvir, Shubhangi Saraf, and Charles Wolf.
\newblock \href {http://dx.doi.org/10.1007/s00454-018-0039-4} {On the Number of
  Ordinary Lines Determined by Sets in Complex Space}.
\newblock {\em Discret. Comput. Geom.}, 61(4):778--808, 2019.

\bibitem[BDWY13]{barak2013fractional}
Boaz Barak, Zeev Dvir, Avi Wigderson, and Amir Yehudayoff.
\newblock {Fractional Sylvester--Gallai theorems}.
\newblock {\em Proceedings of the National Academy of Sciences},
  110(48):19213--19219, 2013.

\bibitem[BM90]{BorwenMoser90}
Peter Borwein and William O.~J. Moser.
\newblock \href {http://dx.doi.org/doi:10.1007/BF02112289} {{A survey of
  Sylvester's problem and its generalizations}}.
\newblock {\em Aequationes Mathematicae}, 40:111--135, 1990.

\bibitem[BMS13]{DBLP:journals/iandc/BeeckenMS13}
Malte Beecken, Johannes Mittmann, and Nitin Saxena.
\newblock \href {http://dx.doi.org/10.1016/j.ic.2012.10.004} {Algebraic
  independence and blackbox identity testing}.
\newblock {\em Inf. Comput.}, 222:2--19, 2013.

\bibitem[Chv04]{DBLP:journals/dcg/Chvatal04}
Vasek Chv{\'{a}}tal.
\newblock \href {http://dx.doi.org/10.1007/s00454-003-0795-6} {Sylvester-Gallai
  Theorem and Metric Betweenness}.
\newblock {\em Discret. Comput. Geom.}, 31(2):175--195, 2004.

\bibitem[CKS18]{DBLP:conf/coco/ChouKS18}
Chi{-}Ning Chou, Mrinal Kumar, and Noam Solomon.
\newblock \href {http://dx.doi.org/10.4230/LIPIcs.CCC.2018.13} {{Hardness vs
  Randomness for Bounded Depth Arithmetic Circuits}}.
\newblock In Rocco~A. Servedio, editor, {\em 33rd Computational Complexity
  Conference, {CCC} 2018, June 22-24, 2018, San Diego, CA, {USA}}, volume 102
  of {\em LIPIcs}, pages 13:1--13:17. Schloss Dagstuhl - Leibniz-Zentrum
  f{\"{u}}r Informatik, 2018.

\bibitem[CLO07]{CLO}
David~A. Cox, John Little, and Donal O'Shea.
\newblock {\em Ideals, Varieties, and Algorithms: An Introduction to
  Computational Algebraic Geometry and Commutative Algebra}.
\newblock Springer, 3rd edition, 2007.

\bibitem[DH16]{DBLP:journals/dcg/DvirH16}
Zeev Dvir and Guangda Hu.
\newblock \href {http://dx.doi.org/10.1007/s00454-016-9781-7}
  {{Sylvester-Gallai for Arrangements of Subspaces}}.
\newblock {\em Discrete {\&} Computational Geometry}, 56(4):940--965, 2016.

\bibitem[DS07]{DBLP:journals/siamcomp/DvirS07}
Zeev Dvir and Amir Shpilka.
\newblock \href {http://dx.doi.org/10.1137/05063605X} {Locally Decodable Codes
  with Two Queries and Polynomial Identity Testing for Depth 3 Circuits}.
\newblock {\em {SIAM} J. Comput.}, 36(5):1404--1434, 2007.

\bibitem[DSW14]{DSW12}
Zeev Dvir, Shubhangi Saraf, and Avi Wigderson.
\newblock \href {http://arxiv.org/abs/1211.0330} {Improved rank bounds for
  design matrices and a new proof of Kelly's theorem}.
\newblock {\em Forum of Mathematics, Sigma}, 2, 2014.
\newblock Pre-print available at \href {http://arxiv.org/abs/1211.0330}
  {\path{arXiv:1211.0330}}.

\bibitem[DSY09]{DBLP:journals/siamcomp/DvirSY09}
Zeev Dvir, Amir Shpilka, and Amir Yehudayoff.
\newblock \href {http://dx.doi.org/10.1137/080735850} {Hardness-Randomness
  Tradeoffs for Bounded Depth Arithmetic Circuits}.
\newblock {\em {SIAM} J. Comput.}, 39(4):1279--1293, 2009.

\bibitem[Dvi12]{DBLP:journals/fttcs/Dvir12}
Zeev Dvir.
\newblock \href {http://dx.doi.org/10.1561/0400000056} {Incidence Theorems and
  Their Applications}.
\newblock {\em Found. Trends Theor. Comput. Sci.}, 6(4):257--393, 2012.

\bibitem[Erd43]{Erdos43}
Paul Erd\"os.
\newblock \href {http://www.jstor.org/stable/2304011} {{Problems for Solution:
  4065}}.
\newblock {\em The American Mathematical Monthly}, 50(1):65, 1943.

\bibitem[FGT19]{DBLP:journals/cacm/FennerGT19}
Stephen~A. Fenner, Rohit Gurjar, and Thomas Thierauf.
\newblock \href {http://dx.doi.org/10.1145/3306208} {A deterministic parallel
  algorithm for bipartite perfect matching}.
\newblock {\em Commun. {ACM}}, 62(3):109--115, 2019.

\bibitem[For14]{ForbesThesis}
Michael~A. Forbes.
\newblock {\em Polynomial identity testing of read-once oblivious algebraic
  branching programs}.
\newblock PhD thesis, Massachusetts Institute of Technology, 2014.

\bibitem[FS13]{DBLP:conf/approx/ForbesS13}
Michael~A. Forbes and Amir Shpilka.
\newblock \href {http://dx.doi.org/10.1007/978-3-642-40328-6_37} {Explicit
  Noether Normalization for Simultaneous Conjugation via Polynomial Identity
  Testing}.
\newblock In Prasad Raghavendra, Sofya Raskhodnikova, Klaus Jansen, and
  Jos{\'{e}} D.~P. Rolim, editors, {\em Approximation, Randomization, and
  Combinatorial Optimization. Algorithms and Techniques - 16th International
  Workshop, {APPROX} 2013, and 17th International Workshop, {RANDOM} 2013,
  Berkeley, CA, USA, August 21-23, 2013. Proceedings}, volume 8096 of {\em
  Lecture Notes in Computer Science}, pages 527--542. Springer, 2013.

\bibitem[FSV18]{DBLP:journals/toc/ForbesSV18}
Michael~A. Forbes, Amir Shpilka, and Ben~Lee Volk.
\newblock \href {http://dx.doi.org/10.4086/toc.2018.v014a018} {{Succinct
  Hitting Sets and Barriers to Proving Lower Bounds for Algebraic Circuits}}.
\newblock {\em Theory of Computing}, 14(1):1--45, 2018.

\bibitem[Gal44]{Galai44}
Tibor Gallai.
\newblock {Solution to Problem 4065}.
\newblock {\em The American Mathematical Monthly}, 51:169--171, 1944.

\bibitem[GKKS16]{KKS16Jour}
Ankit Gupta, Pritish Kamath, Neeraj Kayal, and Ramprasad Saptharishi.
\newblock \href {http://dx.doi.org/10.1137/140957123} {Arithmetic Circuits: {A}
  Chasm at Depth 3}.
\newblock {\em {SIAM} J. Comput.}, 45(3):1064--1079, 2016.

\bibitem[GKSS17]{DBLP:journals/corr/Grochow0SS17}
Joshua~A. Grochow, Mrinal Kumar, Michael~E. Saks, and Shubhangi Saraf.
\newblock \href {http://arxiv.org/abs/1701.01717} {Towards an algebraic natural
  proofs barrier via polynomial identity testing}.
\newblock {\em CoRR}, abs/1701.01717, 2017.
\newblock Pre-print available at \href {http://arxiv.org/abs/1701.01717}
  {\path{arXiv:1701.01717}}.

\bibitem[GOS21]{GOS}
Abhibhav Garg, Rafael Oliviera, and Akash~Kumar Sengupta.
\newblock {Robust Radical Sylvester-Gallai Theorem for Quadratics}.
\newblock Personal communication, 2021.

\bibitem[Gup14]{Gupta14}
Ankit Gupta.
\newblock \href {http://eccc.hpi-web.de/report/2014/130} {{Algebraic Geometric
  Techniques for Depth-4 {PIT} {\&} Sylvester-Gallai Conjectures for
  Varieties}}.
\newblock {\em Electronic Colloquium on Computational Complexity {(ECCC)}},
  21:130, 2014.

\bibitem[Han65]{Hansen65}
Sten Hansen.
\newblock {A generalization of a theorem of Sylvester on the lines determined
  by a finite point set}.
\newblock {\em Mathematica Scandinavica}, 16:175--180, 1965.

\bibitem[HS80]{DBLP:conf/stoc/HeintzS80}
Joos Heintz and Claus{-}Peter Schnorr.
\newblock \href {http://dx.doi.org/10.1145/800141.804674} {Testing Polynomials
  which Are Easy to Compute (Extended Abstract)}.
\newblock In Raymond~E. Miller, Seymour Ginsburg, Walter~A. Burkhard, and
  Richard~J. Lipton, editors, {\em Proceedings of the 12th Annual {ACM}
  Symposium on Theory of Computing, April 28-30, 1980, Los Angeles, California,
  {USA}}, pages 262--272. {ACM}, 1980.

\bibitem[Kel86]{Kelly86}
Leroy~Milton Kelly.
\newblock {A resolution of the Sylvester-Gallai problem of J.-P. Serre}.
\newblock {\em Discrete \& Computational Geometry}, 1(2):101--104, 1986.

\bibitem[KI04]{DBLP:journals/cc/KabanetsI04}
Valentine Kabanets and Russell Impagliazzo.
\newblock \href {http://dx.doi.org/10.1007/s00037-004-0182-6} {Derandomizing
  Polynomial Identity Tests Means Proving Circuit Lower Bounds}.
\newblock {\em Computational Complexity}, 13(1-2):1--46, 2004.

\bibitem[KS09a]{DBLP:conf/coco/KarninS09}
Zohar~S. Karnin and Amir Shpilka.
\newblock \href {http://dx.doi.org/10.1109/CCC.2009.18} {Reconstruction of
  Generalized Depth-3 Arithmetic Circuits with Bounded Top Fan-in}.
\newblock In {\em Proceedings of the 24th Annual {IEEE} Conference on
  Computational Complexity, {CCC} 2009, Paris, France, 15-18 July 2009}, pages
  274--285. {IEEE} Computer Society, 2009.

\bibitem[KS09b]{DBLP:conf/focs/KayalS09}
Neeraj Kayal and Shubhangi Saraf.
\newblock \href {http://dx.doi.org/10.1109/FOCS.2009.67} {Blackbox Polynomial
  Identity Testing for Depth 3 Circuits}.
\newblock In {\em 50th Annual {IEEE} Symposium on Foundations of Computer
  Science, {FOCS} 2009, October 25-27, 2009, Atlanta, Georgia, {USA}}, pages
  198--207. {IEEE} Computer Society, 2009.

\bibitem[KS19]{DBLP:journals/eatcs/0001S19}
Mrinal Kumar and Ramprasad Saptharishi.
\newblock \href
  {http://bulletin.eatcs.org/index.php/beatcs/article/view/591/599}
  {{{Hardness-Randomness} Tradeoffs for Algebraic Computation}}.
\newblock {\em Bull. {EATCS}}, 129, 2019.

\bibitem[KSS15]{DBLP:journals/cc/KoppartySS15}
Swastik Kopparty, Shubhangi Saraf, and Amir Shpilka.
\newblock \href {http://dx.doi.org/10.1007/s00037-015-0102-y} {Equivalence of
  Polynomial Identity Testing and Polynomial Factorization}.
\newblock {\em Computational Complexity}, 24(2):295--331, 2015.

\bibitem[Mel41]{Melch41}
Eberhard Melchior.
\newblock {\"Uber Vielseite der Projektive Ebene}.
\newblock {\em Deutsche Mathematik}, 5:461–475, 1941.

\bibitem[Mul17]{Mulmuley-GCT-V}
Ketan~D. Mulmuley.
\newblock {Geometric complexity theory V: Efficient algorithms for Noether
  normalization}.
\newblock {\em J. Amer. Math. Soc.}, 30(1):225--309, 2017.

\bibitem[PS09]{DBLP:journals/dm/PretoriusS09}
Lourens~M. Pretorius and Konrad~J. Swanepoel.
\newblock \href {http://dx.doi.org/10.1016/j.disc.2007.12.027} {The
  Sylvester-Gallai theorem, colourings and algebra}.
\newblock {\em Discret. Math.}, 309(2):385--399, 2009.

\bibitem[PS20a]{Peleg-Shpilka-SG}
Shir Peleg and Amir Shpilka.
\newblock \href {http://dx.doi.org/10.4230/LIPIcs.CCC.2020.8} {{A Generalized
  Sylvester-Gallai Type Theorem for Quadratic Polynomials}}.
\newblock In Shubhangi Saraf, editor, {\em 35th Computational Complexity
  Conference, {CCC} 2020, July 28-31, 2020, Saarbr{\"{u}}cken, Germany (Virtual
  Conference)}, volume 169 of {\em LIPIcs}, pages 8:1--8:33. Schloss Dagstuhl -
  Leibniz-Zentrum f{\"{u}}r Informatik, 2020.

\bibitem[PS20b]{Peleg-Shpilka-PIT}
Shir Peleg and Amir Shpilka.
\newblock \href {https://arxiv.org/abs/2006.08263} {Polynomial time
  deterministic identity testing algorithm for
  {\(\Sigma\)}\({}^{\mbox{[3]}}\){\(\Pi\)}{\(\Sigma\)}{\(\Pi\)}\({}^{\mbox{[2]}}\)
  circuits via Edelstein-Kelly type theorem for quadratic polynomials}.
\newblock {\em CoRR}, abs/2006.08263, 2020.
\newblock Pre-print available at \href {http://arxiv.org/abs/2006.08263}
  {\path{arXiv:2006.08263}}.

\bibitem[Sax09]{Saxena09}
Nitin Saxena.
\newblock \href {https://eccc.weizmann.ac.il/report/2009/101/} {Progress on
  polynomial identity testing}.
\newblock {\em Bulletin of EATCS}, 99:49--79, 2009.

\bibitem[Sax14]{Saxena14}
Nitin Saxena.
\newblock \href {https://books.google.co.il/books?id=U7ApBAAAQBAJ} {Progress on
  Polynomial Identity Testing-II}.
\newblock In M.~Agrawal and V.~Arvind, editors, {\em Perspectives in
  Computational Complexity: The Somenath Biswas Anniversary Volume}, Progress
  in Computer Science and Applied Logic, pages 131--146. Springer International
  Publishing, 2014.

\bibitem[Ser66]{Serre66}
Jean-Pierre Serre.
\newblock \href {http://www.jstor.org/stable/2313941} {{Advanced Problems:
  5359}}.
\newblock {\em The American Mathematical Monthly}, 73(1):89, 1966.

\bibitem[Shp09]{DBLP:journals/siamcomp/Shpilka09}
Amir Shpilka.
\newblock \href {http://dx.doi.org/10.1137/070694879} {Interpolation of Depth-3
  Arithmetic Circuits with Two Multiplication Gates}.
\newblock {\em {SIAM} J. Comput.}, 38(6):2130--2161, 2009.

\bibitem[Shp20]{ShpilkaSG}
Amir Shpilka.
\newblock {Sylvester-Gallai} type theorems for quadratic polynomials.
\newblock {\em Discrete Analysis}, 13, 2020.

\bibitem[Sin16]{DBLP:conf/coco/Sinha16}
Gaurav Sinha.
\newblock \href {http://dx.doi.org/10.4230/LIPIcs.CCC.2016.31} {{Reconstruction
  of Real Depth-3 Circuits with Top Fan-In 2}}.
\newblock In Ran Raz, editor, {\em 31st Conference on Computational Complexity,
  {CCC} 2016, May 29 to June 1, 2016, Tokyo, Japan}, volume~50 of {\em LIPIcs},
  pages 31:1--31:53. Schloss Dagstuhl - Leibniz-Zentrum fuer Informatik, 2016.

\bibitem[SS12]{DBLP:journals/siamcomp/SaxenaS12}
Nitin Saxena and Comandur Seshadhri.
\newblock \href {http://dx.doi.org/10.1137/10848232} {Blackbox Identity Testing
  for Bounded Top-Fanin Depth-3 Circuits: The Field Doesn't Matter}.
\newblock {\em {SIAM} J. Comput.}, 41(5):1285--1298, 2012.

\bibitem[SS13]{DBLP:journals/jacm/SaxenaS13}
Nitin Saxena and Comandur Seshadhri.
\newblock \href {http://dx.doi.org/10.1145/2528403} {From Sylvester-Gallai
  configurations to rank bounds: Improved blackbox identity test for depth-3
  circuits}.
\newblock {\em J. {ACM}}, 60(5):33, 2013.

\bibitem[ST17]{DBLP:conf/focs/SvenssonT17}
Ola Svensson and Jakub Tarnawski.
\newblock \href {http://dx.doi.org/10.1109/FOCS.2017.70} {{The Matching Problem
  in General Graphs Is in Quasi-NC}}.
\newblock In Chris Umans, editor, {\em 58th {IEEE} Annual Symposium on
  Foundations of Computer Science, {FOCS} 2017, Berkeley, CA, USA, October
  15-17, 2017}, pages 696--707. {IEEE} Computer Society, 2017.

\bibitem[SY10]{DBLP:journals/fttcs/ShpilkaY10}
Amir Shpilka and Amir Yehudayoff.
\newblock \href {http://dx.doi.org/10.1561/0400000039} {{Arithmetic Circuits:
  {A} survey of recent results and open questions}}.
\newblock {\em Foundations and Trends in Theoretical Computer Science},
  5(3-4):207--388, 2010.

\bibitem[Syl93]{Syl1893}
James~Joseph Sylvester.
\newblock {Mathematical question 11851}.
\newblock {\em Educational Times}, pages 59--98, 1893.

\end{thebibliography}

\appendix

\section{Appendix}\label{sec:appendix}

\begin{proof}[Proof of \autoref{thm:lin-rob-space}]
	We say that $q \in \cT$ is a neighbor of  $p\in \cK$, if the span of $p$ and $q$ contains a third point in $\cT$. We denote with $\Gamma(p)$ the set of all neighbors of $p$. 
	
	We perform the following process:  While there is $p\in \cK$ such that $|\Gamma(p)\cap W|\geq 0.1\cdot \delta|\cT|$ we set, abusing notation, $W=\spn{W,p}$,  $\cW = \cT\cap W$ and $\cK = \cT\setminus \cW$. We first wish to prove that this process must terminate after $O(1/\delta)$ many iterations. For that we will show that at every step of the process, at least  $0.1\cdot \delta |\cK|$ points from $\cK$ are moved to $\cW$.
	
	So consider a step of the process, and let $p\in \cK$ be the relevant polynomial. Observe that for any $w \in \Gamma(p)\cap W$, the space spanned by $p$ and $w$ must contain a point in $\cK$, as at every step of the process we maintain that $\cK\cap W=\emptyset$.
	Furthermore, all the points in $\cK$
	that are obtained in this manner (i.e. that are in the span of $p$ and some $w\in\cW$) must be distinct. Indeed, if $p$ spans $q \in \cK$ with two different 	$w_1, w_2 \in \cW$ then, as $w_1$ and $w_2$ are linearly independent, and so are $p$ and $q$, we get that
	$\spn{p, q} = \spn{w_1, w_2}$ which implies that $p \in W$, in contradiction. From this and the assumption that $|\Gamma(p)\cap W|\geq 0.1\cdot \delta|\cT|\geq 0.1\cdot \delta |\cK|$, it follows that at least $0.1\cdot \delta |\cK|$ points from $\cK$ are moved to $\cW$ at this step of our process.
	Hence, the process must terminate after at most $10/\delta$ steps. Note that when the process terminates, $\dim(W)\leq r + \frac{10}{\delta}$. 
	
	After the termination of the process, it holds that for every $p\in \cK$ ,$|\Gamma(p)\cap W|< 0.1\cdot \delta|\cT|$. We now wish to prove that $\cK$ satisfies the conditions of \autoref{thm:DSW} (the usual robust-SG theorem ) with parameter $\delta'=0.8\delta$. For this we have to prove that  for every $p\in \cK$ there are at least $0.8\cdot \delta |\cK|$ points in $\cK$ that each spans with $p$ a third point in $\cK$. To show this we prove that the number of points $q\in \Gamma(p)$ that do not span with $p$ a point in $\cK$, is at most $0.1\cdot \delta|\cK|$.
	
	Indeed, if $q\in \Gamma(p)\cap \cK$ does not span with $p$ a point in $\cK$ then it must span some $w\in \cW$.  Furthermore, $w\in\Gamma(p)$. The crucial observation is that there cannot be any other point $q\neq q'\in \cK\cap \Gamma(p)$ such that $w\in \spn{p,q}\cap  \spn{p,q'}$, as in this case, by our assumption on pairwise independence, we would have that $q'\in\spn{p,q}$, in contradiction to the assumption that $q$ does not span with $p$ a point in $\cK$. Therefore, when restricting ourselves to $\cK$, we have that $p$ has at least $\delta|\cT| - 2|\Gamma(p)\cap W|\geq  0.8\cdot \delta|\cT| \geq  0.8\cdot \delta|\cK|$ neighbors in $\cK$ that span together with $p$ a third point in $\cK$, as we wanted to prove.  
	
	As we just proved that  $\cK$ satisfies \autoref{thm:DSW} with parameter $\delta'=0.8\delta$, it follows that $\dim(\cK)\leq \frac{15}{\delta}+1$. In conclusion we get that  $\dim(\cT)\leq \dim(W)+\dim(\cK)\leq \left(r +\frac{10}{\delta}\right) +\left(\frac{15}{\delta}+1\right)= O(r + \frac{1}{\delta})$, as claimed.
\end{proof}

\begin{proof}[Proof of \autoref{cla:case-1-cut}]
	For a point $v\in \cB$ denote by $\Gamma(v) = \{u\in V\setminus \cB\mid u,v \text{ lie on a special line}\}$, similarly for a point $u\in V\setminus \cB$ denote by $\Gamma(u) = \{v\in\cB\mid v,u \text{ lie on a special line}\}$. Note that for $v\in \cB$, $u\in \Gamma(v)$ if the special line between $v$ and $u$ contains $w\in V\setminus \cB$, then $w\in \Gamma(V)$. 	
	Denote \[\overline{\gamma}(V) = \frac{\sum_{v\in V} \card{\Gamma(v)}}{\card{V}} .\]
	From the assumption that there are $\delta m^2$ pairs of points $v\in \cB$, $u\in V\setminus \cB$ that  lie on a special line it follows that $\overline{\gamma}(V) \geq 2\delta\cdot m$.
	If for every $v \in V$ it holds that $|\Gamma(v)|\geq(\delta/2)\cdot m$ then $V$ is a $\delta/2$-SG configuration and thus $\dim(V) = O(1/\delta)$ and in particular $\cB' = \cB$ satisfies the claim.  
	
	While there is $v \in V$ with $\card{\Gamma(v)}<(\delta/2)\cdot m$, remove $v$ from $V$. We want to bound the number of lines that became ordinary after removing $v$. Assume without loss of generality, that $v\in \cB$.
	
	First, we remove one neighbor for every point in $\Gamma (v)\subset V \setminus\cB$. The other case to consider is two neighbors $v'\in \cB, u\in V\setminus \cB$ that lie on a special line containing $v$, and thus, it might be that after removing $v$ they do not lie on a special line anymore. In this case we should remove $u$ from $\Gamma(v')$ and vice versa. Note that in this case we have $u\in \Gamma(v)$. If there is another $v''\in \cB$ such that $v$ lie on the line between $u$ and $v''$, then the line between $u$ and $v$ contains both $v'$ and $v''$. in particular, removing $v$ still keeps the line special as it intersect the set with at least three points $u,v',v''$. Therefore, we do not need to remove any neighbors of $u$. To conclude, we removed at most two lines for every point in $\Gamma(v)$. As each line that became ordinary affects the neighborhoods of the (only) two points that are on it. Thus, we conclude that,

	\[\overline{\gamma}_{new}(V) = \frac{\sum_{u\in V} |\Gamma(u)|-4|\Gamma(v)|} {n-1} > \overline{\gamma}(V).\]
	
	as $4|\Gamma(v)| < 2\delta\cdot m.$ Thus we can continue this process until all the points satisfy $|\Gamma(v)|\geq(\delta/2)\cdot m$, as the average degree increases in every step, the process must terminates and leave $|\cB| >(\delta/2) \cdot m$,  and thus $\cB' = \cB$ satisfies the claim.
	
\end{proof}

\end{document}